\documentclass[11pt]{article}
\usepackage[top=1in, bottom=1in, left=1in, right=1in]{geometry}
\usepackage{graphicx} 
\usepackage[numbers]{natbib}
\usepackage{amsmath}
\usepackage{amsthm}
\usepackage{amsfonts}
\usepackage{tikz}

\usepackage{amssymb}
\usepackage{mathtools}
\usepackage{bm}
\usepackage{tikz}
\usepackage{hyperref}
\usepackage{float}
\usepackage{cleveref}
\usepackage{xcolor}
\usepackage[shortlabels]{enumitem}
\usepackage{xspace}
\usepackage{array}
\usepackage{hyperref}
\usepackage{booktabs}
\usepackage{subcaption} 
\usepackage{graphicx}   
\usepackage{makecell} 
\usepackage{soul}

\newif\ifshowshcomments
\showshcommentstrue

\newtheorem{claim}{Claim}
\newtheorem{definition}{Definition}
\newtheorem{theorem}{Theorem}
\newtheorem{proposition}{Proposition}
\newtheorem{observation}{Observation}
\newtheorem{lemma}{Lemma}
\newtheorem{regularity assumption}{Regularity assumption}

\newtheorem{corollary}{Corollary}
\newtheorem{example}{Example}

\newcommand{\context}{\text{framing} }

\DeclareMathOperator*{\argmax}{arg\,max}
\DeclareMathOperator*{\argmin}{arg\,min}
\DeclareMathOperator*{\maximize}{maximize}
\newcommand{\eps}{\varepsilon}

\newcommand{\BS}{\text{BS}}
\newcommand{\Ss}{\mathcal{S}}
\newcommand{\I}{\mathcal{I}}
\newcommand{\OC}{\text{OF}}
\newcommand{\IBS}{\mathcal{I}_{BS}}
\newcommand{\IOC}{\mathcal{I}_{OF}}
\newcommand{\al}{a_{\ell}}

\newcommand{\atheta}{a^{\theta}}
\newcommand{\A}{\mathcal{A}}
\newcommand{\muhat}{\hat{\mu}}
\newcommand{\wtilde}{\widetilde{\omega}}
\newcommand{\atilde}{\widetilde{a}}
\newcommand{\vmax}{v^{max}}
\newcommand{\E}{\mathbb{E}}

\newcommand{\squishlist}{
   \begin{list}{$\bullet$}
    { \setlength{\itemsep}{0pt}      \setlength{\parsep}{0pt}
      \setlength{\topsep}{3pt}       \setlength{\partopsep}{0pt}
      \setlength{\leftmargin}{1.5em} \setlength{\labelwidth}{1em}
      \setlength{\labelsep}{0.5em} } }
\newcommand{\squishend}{  \end{list}  }

\title{
Information Design with Large Language Models%
\thanks{
This work features both a theoretical framework, and empirical results. For the experiments, we use GPT-5.2~\cite{agarwal2025gpt} and Claude Sonnet 4~\cite{anthropic2025claudesonnet4} for the LLM portions of the framework, and SciPy~\cite{2020SciPy-NMeth} for the analytical parts. All experiments were conducted by the Harvard University group.}}
\author
{Paul D\"{u}tting \\ Google Research
\and
Safwan Hossain\footnote{Authors listed alphabetically. Correspondence to \href{mailto:shossain@g.harvard.edu}{\texttt{shossain@g.harvard.edu}} or \href{mailto:shossain@g.harvard.edu}{\texttt{lintao@cuhk.edu.cn}}} \\ Harvard University
\and
Tao Lin$^{\text{†}}$ \\ Harvard University
\and
Renato Paes Leme \\ Google Research
\and
Sai Srivatsa Ravindranath \\ Harvard University
\and
Haifeng Xu \\ University of Chicago
\and
Song Zuo \\ Google Research}
\date{}

\begin{document}

\maketitle

\begin{abstract}
Information design is typically studied through the lens of Bayesian signaling, where signals shape beliefs purely based on their correlation with the true state of the world. However, behavioral economics and psychology emphasize that human decision-making is more complex and can depend on how information is framed. This paper formalizes a language-based notion of framing and bridges this to the popular Bayesian-persuasion model. We model framing as a possibly non-Bayesian, linguistic way to influence a receiver's prior belief, while a signaling/recommendation scheme can further refine this belief in the classic Bayesian way. A key challenge in systematically optimizing in this framework is the vast space of possible framings and the difficulty of predicting their effects on receivers. Based on growing evidence that Large Language Models (LLMs) can effectively serve as proxies for human behavior, we formulate a theoretical model based on access to a framing-to-belief mapping. This model then enables us to precisely characterize when solely optimizing framing or jointly optimizing framing and signaling is tractable. We substantiate our theoretical analysis with an empirical study that leverages LLMs to optimize over the natural-language framing space using an iterative prompt optimization method combined with analytical solvers for optimal signaling schemes.
\end{abstract}


\section{Introduction}

Information design is a core concept in microeconomics and decision theory that considers the strategic communication of information from one party (i.e., \emph{sender}) to shape the decisions of another (i.e., \emph{receiver}) \citep{bergemann2019information}. Classical information design models, such as \emph{cheap talk} \citep{crawford_strategic_1982} and \emph{Bayesian persuasion} \citep{kamenica2011bayesian}, assume that the sender can choose and reveal some correlation between the relevant \emph{states} of the world and communicated information (\emph{signal}), allowing the decision-maker to update their belief in a fully Bayesian manner. That is, the chosen quantitative correlations between states and signals are all that is relevant in influencing beliefs in such theoretical models. 

This purely Bayesian perspective, however, stands in contrast to established findings in behavioral economics and psychology, which emphasize that human decision-making is more complex and can depend on the \emph{framing} of information. The landmark work of \citet{tversky1981framing} observes that humans, asked to choose between two options to combat the outbreak of a hypothetical disease, decide differently depending on how the downstream effects of the two options are phrased (e.g., ``$30\%$ change of saving 600 people'' \emph{vs} ``70\% chance that all people will die''). To explain this seemingly irrational behavior, they argue that raw information combines with societal norms to affect the decision-maker's perception of ``acts, outcomes, and contingencies''. Closer to our setting, there is ample evidence that ``persuasion in the field'' is much richer than the existing information design models (see, e.g., the influential survey by \citet{dellavigna2010persuasion}), where (i) tangential information given before a recommendation (signal); (ii) text or phrasing used to support a recommendation; (iii) and even non-textual cues such as color, background images, or fonts, all impact the effectiveness of persuasive communication. 

To make this concrete and outline a running example, consider an advertising scenario where a clothing brand designs the slogan for a new line of products to target a group of consumers (e.g., outdoor enthusiasts or fashion forward youth). The slogan (framing) typically contains generic platitudes or information that applies to all products carried by the brand. Different slogans, which can emphasize different aspects of the brand or use specific phrasing, induce different perceptions. This mapping from slogans to perception is consumer-specific, possibly non-Bayesian and importantly, influenced by a myriad of factors that the sender cannot influence. Alongside designing a slogan (which is common for all products in this line), the clothing brand can offer different discounts for different products. We view discounts as state-dependent signals because a consumer can observe using historical trends how strongly a large discount correlates with product quality or style, and update their belief about a product accordingly. This correlation between discounts and product features is directly determined by the brand. From the clothing brand's perspective then, it may be free to optimize over both the slogan (framing) and the discounting strategy (signaling) to better persuade 
its target demographic. 

Systematically reasoning about and optimizing over 
framing
is generally challenging.  
First, the space of all framings to choose from is a vast, primarily linguistic space. 
Second, the effect of framing on a receiver is personalized and difficult to characterize rigorously. 
Fortunately, the advent of Large Language Models (LLMs) offers a new approach to addressing these twin challenges. LLMs are adept at understanding the structure of natural language. Furthermore, a growing line of work in economics (e.g., \citep{horton2023large, brand2023using, leng2024can, dillion2023can, manning2024automated}) shows that LLMs can be used to simulate human behavior 
on a host of problems, or even tuned to be agents that 
make decisions on behalf of humans \citep{fish2023generative, soumalias2025llmpowered}. Thus, it it timely to initiate a formal exploration of these challenges within the context of information design. 
We propose a theoretical model for information design with framing effect, analyze different variants of framing optimization problems under this model, and employ LLMs to explore and exploit framing effects in information design.

\subsection{Summary of Our Contributions}
\textbf{Model and Problem Formulation:} Our first contribution is a theoretical framework for information design with framing effect (Section~\ref{sec:model}) that builds on the classical \emph{Bayesian Persuasion} model \citep{kamenica2011bayesian}.
There is a sender (she) and a receiver (he). A state of the world $\omega \in \Omega$ is drawn from a distribution $\mu_0 \in \Delta(\Omega)$, with the sender observes $\omega$ and the receiver not.
The receiver has a set of actions $\A$ that affect the possibly misaligned utilities of the two players. He chooses an action $a^\star \in \A$ to maximize his expected utility given his belief.
The sender seeks to influence the receiver's action by influencing his belief over the states. 

We define a ``framing'' as a state-independent message $c$ 
in
some abstract space $C$ (in the advertising example, this is the set of all possible slogans) that affects the receiver's belief in a non-Bayesian way. The sender chooses a framing $c \in C$ before observing the state, which induces/primes the receiver to form a \emph{subjective prior belief} $\mu_c \in B \subseteq \Delta(\Omega)$ about the state, where $B$ is the space of inducible beliefs. The mapping $\ell: c \mapsto \mu_c$ from framing to induced belief is determined by nature or societal norms, not by the sender. Our approach of modeling framing as affecting belief formation, rather than perception of payoffs, is supported by the framing effect literature~\citep{ellingsen2012social, druckman2001limits}. In addition to choosing a framing $c$ to shape the receiver's prior belief $\mu_c$, the sender can also commit to a signaling scheme $\pi : \Omega \to \Delta(S)$, which is a possibly randomized mapping from states to signals. By communicating a realized signal $s$ upon observing $\omega$, the sender further refines the receiver's prior belief $\mu_c$ to a posterior $\mu_{c, s}$ via Bayes rule. As in classic Bayesian persuasion, the receiver takes an expected utility maximizing action based on his posterior belief. This model gives rise to three possible optimization problems for the sender: 
 \squishlist
    \item[{(a)}]{\em Signaling-only:} The framing $c$ (and thus the receiver's prior belief $\mu_c$) is fixed; the sender only optimizes the signaling scheme $\pi$.
    \item[{(b)}]{\em Framing-only:} The sender's signaling scheme $\pi$ is fixed; she only optimizes the framing $c$.
    \item[{(c)}]{\em Framing-signaling joint optimization:} The sender jointly optimizes both framing $c$ and signaling scheme $\pi$. 
\squishend
The first setting corresponds to facing a receiver with a fixed but possibly distinct (from the sender) prior belief. This direction has been studied by the literature on Bayesian persuasion with heterogeneuous priors (e.g., \citep{alonso_bayesian_2016}). The key thrust of our work is thus to study settings (b) and (c), both theoretically and empirically. Setting (b) is relevant when the sender is bound to an information revelation scheme they have committed to in the past (e.g., a multi-year advertising strategy or regulation restrictions), but can change the framing (e.g., endorsement, wording). Setting (c) allows the sender to choose both with full freedom. 

\textbf{Theoretical Results:}
Our theoretical investigations take an optimistic perspective that LLMs can (1) simulate framing effects and (2) search in language space. As such, we primarily consider the easier problem of directly optimizing over the numerical framing induced belief space $B$, as opposed to the natural language framing space $C$. Within this setting, \Cref{sec:context_only_optimization} investigates setting (b), optimizing framing for a fixed signaling scheme, while \Cref{sec:joint} studies setting (c), jointly optimizing framing and signaling scheme. 

Our results reveal that setting (b) poses severe challenges, even under this generous LLM-inspired setting. Specifically, we show that the sender's expected utility will generally be a discontinuous function of the receiver's framing induced prior belief $\mu_c$ (Proposition~\ref{prop:fixed-scheme-discontinuous}). This means that even slight errors in the mapping between framing and belief can have significant implications. We also show that the computation of the approximately optimal framing induced belief under a fixed signaling scheme is NP-hard (Theorem~\ref{thrm:np_hardness}). In contrast, setting (c), jointly optimizing framing and signaling scheme, is more tractable. In particular, we show that the sender's expected utility, as a function of the framing-induced pior belief $\mu_c$ and the corresponding optimal signaling scheme $\pi^*_c$, is continuous (Theorem~\ref{thrm:joint_continuous}), making the setting graceful to errors. We also provide a quasi-polynomial-time approximation scheme (QPTAS) for the joint optimization problem (Theorem~\ref{thm:qptas}), with specific instances admitting better results (Theorem~\ref{thm:joint-unconstrained}). These findings support the conclusion that the joint optimization problem tends to be the more tractable problem for systematic optimization, inspiring our experiments. Overall, our work characterizes the unique structural and optimization landscape of this generalized model of information design. 

\textbf{Empirical Results:}
Section~\ref{sec:experiment} shifts the work from a theoretical exploration to a practical one where we study the systematic co-design of natural language framing and signaling scheme. LLMs are used to (1) approximate how receivers would perceive a given framing and (2) directly search over the language-base framing space $C$ and improve based on feedback. Combining this with analytical methods for  computing the optimal signaling scheme given a belief yields an end-to-end optimization framework. We demonstrate the efficacy of our approach through a case-study based on the advertising example discussed above. The results include preliminary explorations on how well LLMs estimate beliefs from framing, and evidence that the proposed framework is effective in generating good framing and signaling candidates.

\subsection{Related Work}
\textbf{Algorithmic information design.}
The algorithmic study of information design 
has attracted significant recent interest. This literature starts from the complexity-theoretic study initiated by \citet{dughmi_algorithmic_2016}, and lately has integrated many aspects of machine learning to address unknowns in the setting \citep{castiglioni2020online, feng2022online, lin_information_2025}.
\citet{haghtalab_communicating_2024} study a communication game where the sender persuades a receiver by providing ``anecdotes'', instead of committing to recommendation strategies, assuming a known specific model of how anecdotes influence the receiver's behavior. While their anecdote effect shares a similar spirit to our framing effect, we do not assume any specific framing effect. The use of LLM allows us to model real-world framing effects that cannot be easily characterized theoretically. 


\textbf{Information design with LLMs.}
Among recent works on the interface of information design and LLMs, \citet{harris2023algorithmic} study a learning-theoretic question about learning an optimal recommendation strategy by querying an LLM that simulates a receiver with unknown prior belief.
This differs from our aim of introducing a new dimension, i.e., framing, to information design.
\citet{li_verbalized_2025} explore how Bayesian persuasion mechanisms can be implemented in natural language,
while \citet{wu_grounded_2025} develop methods for generating persuasive marketing content that is both linguistically coherent and grounded in product attributes.
Both works share our emphasis on bridging formal persuasion models with natural language generation, but differ in the focus: while they design practical systems for persuasive text generation, our work provides a theoretical framework that explicitly integrates framing into information design and studies its tractability when combined with signaling.

\textbf{Additional evidence for framing effects.}
In addition to the landmark study of \citet{tversky1981framing}, there are numerous studies in Behavioral Economics and Psychology that provide evidence for framing effects.
For instance, \citet{ellingsen2012social} show that naming the standard prisoner's dilemma game differently (e.g., as the ``Community Game'' or ``Stock Market Game'') will lead to different players behaviors despite playing the same underlying game. These behavioral studies are consistent with   ``the hypothesis that social frames are coordination devices... enter people's beliefs rather than their preferences''\citep{ellingsen2012social}.
\citet{nelson1999issue} observe similar belief influence by framing in persuasion problems. These behavioral studies motivate our principled study of using framing in information design. 

\textbf{LLMs as proxys for human behavior.}
Our work subscribes to the recent studies that use LLMs as proxies of human agents to understand the social norms they carry and/or the economic decisions they make \citep{horton2023large, brand2023using, leng2024can, dillion2023can}. Our work 
fits this general theme, but is different from these works in research questions.

\section{Model: Information Design with Framing Effect}\label{sec:model}
\paragraph{\bf Bayesian Persuasion and Signaling.} A fundamental model of information design is the classic Bayesian persuasion problem between two rational players: a \emph{sender} (persuader, with she/her pronouns) and a \emph{receiver} (decision maker, with he/him pronouns). Both players' utilities depend on a \emph{state of the world} $\omega \in \Omega$.
The receiver must take an action $a$ from some finite set $\A$, which, along with the state $\omega$, jointly determines the utilities of both players.\footnote{We use 0 index for actions and states. That is, $\A = \{a_0, \dots, a_{|\A| - 1}\}$, and $\Omega = \{\omega_0, \dots, \omega_{|\Omega| - 1}\}$.}
Formally, the sender's utility function is $u: \A \times \Omega \to \mathbb{R}$ and the receiver's utility function is $v: \A \times \Omega \to \mathbb{R}$. 
The sender does not act, but possesses an information advantage: she is assumed to perfectly observe the realized state of the world $\omega$, whereas the receiver only knows the prior distribution $\mu_0 \in \Delta(\Omega)$ of the state. 

For simplicity, we make a few mild \emph{regularity} assumptions. First, $\mu_0$ has full support, i.e., $\mu_0(\omega) > 0$ for every state $\omega \in \Omega$.
Second, every action $a \in \A$ is strictly inducible: that is, for every $a \in \A$, there is some belief $\mu \in \Delta(\Omega)$ wherein action $a$ is the unique best response action for the receiver, i.e., $\E_{\omega \sim \mu}[v(\omega, a)] > \E_{\omega \sim \mu}[v(\omega, a')]$ for all $a' \in \A\setminus \{a\}$.\footnote{If this assumption does not hold, then action $a$ is weakly dominated and can be removed from the receiver's action set. }
Third, some of our computational results are about additive approximation, which requires players' utilities to be bounded. Hence without loss of generality we  assume both players' utilities $u,v$ are normalized to be within $[0,1]$.

The standard sender-receiver interaction in Bayesian persuasion is through a sender-designed  \emph{signaling scheme} that partially reveal the realized state $\omega$ to the receiver. Formally, for a signaling space $\Ss$, the sender designs and commits to a randomized mapping $\pi: \Omega \to \Delta(\Ss)$, where $\pi(s|\omega)$ specifies the probability of sending signal $s \in \Ss$ when the realized state is $\omega$. The receiver, upon observing a signal $s$ sampled from this signaling scheme, updates his belief over $\Omega$ and takes an expected-utility-maximizing action. As a textbook assumption in this literature, ties are assumed to be broken in favor of the sender.
In the special case where the signaling space $\mathcal S$ equals $\A$, the signaling scheme $\pi: \Omega \to \Delta(\A)$ is also called a \emph{recommendation scheme}, where signals correspond to action recommendations. A recommendation scheme $\pi$ is \emph{obedient} if whenever $\pi$ recommends an action $a \in \A$ to the receiver, $a$ is indeed optimal based on the receiver's posterior belief. We will consider general signaling schemes as well as recommendation schemes in this work.

\paragraph{\bf Framing.}
We define \emph{framing} as a message that affects the receiver's belief in a possibly non-Bayesian way. For example, framing can correspond to natural language phrases that describe the sender, the signaling scheme, or the signal.
The sender selects this framing prior to state observation and it is thus state-independent from the sender's perspective. As an aside, we note that a \emph{state-dependent} framing is a strange object. A sender-determined state-to-framing correlation is (1) functionally no different than signaling and (2) may be inconsistent with the behavioral implications of that framing (e.g. correlating the color red (framing) with feeling full despite this color being known to induce sensations of hunger).


Formally, let $C$ be an abstract space of all possible framings, which may depend on the problem instance $\I$. To connect framing to belief, we assume that each framing $c \in C$ induces a \emph{subjective prior belief} $\mu_c \in \Delta(\Omega)$ of the receiver about the states.
The mapping from $c$ to $\mu_c$ is captured by a function $\ell: C \rightarrow B$, with $\mu_c = \ell(c)$, 
where $B = \{\ell(c) : c\in C \}$ is the set of all inducible prior beliefs.
Importantly, the sender cannot influence this mapping $\ell$ and it is instead dictated by personal or societal norms.
This mapping $\ell$ abstracts out the belief update procedure, which could be Bayesian or non-Bayesian. From a non-Bayesian perspective, the framing may simply form some receiver belief inherited from common sense in human languages.
To exposit the Bayesian perspective, one can imagine that the receiver has some initial belief about the state $\omega$, which was then Bayesian updated to $\mu_c$ based on certain ``societally consensed''  signaling scheme $\sigma(c | \omega)$ after observing the framing $c$. In any case, how the mapping $\ell$ was formed is not concerned in the mathematics of our model.
Instead, in Section \ref{sec:experiment}, we will empirically demonstrate that such a framing-to-belief mapping $\ell$ can be robustly simulated by LLMs, serving as a justification for this modeling primitive.
In Sections \ref{sec:context_only_optimization} and \ref{sec:joint}, we will also analyze the robustness of the optimization problem to the inaccuracy of a simulated framing-to-belief mapping, 
where we allow the simulated mapping, denoted by $\ell_\eps$, to have a small error
in the sense of $|\ell_\epsilon(c) - \ell(c)| \leq \epsilon$ for every $c \in C$. 


A natural-language framing space $C$ is discrete by nature, though it is enormous. In this case,  the corresponding inducible prior belief set $B$ is also discrete. We will also consider the relaxation to continuous framing space, 
in which case the inducible prior belief set $B$ is assumed to be a \emph{convex} subset of the simplex $\Delta(\Omega)$. That is, if two beliefs $\mu_1, \mu_2$ can be induced by some framings,
then there is a framing to induce any belief in between. 
It is valuable to consider the continuous relaxation for two reasons. First, the geometry of the inducible belief set $B$ enables more structural characterization of the optimal design. Second, while this may not precisely map to the ultimately discrete linguistic framing space in practice, the richness of human language renders the continuous approximation reasonable. 
We formalize these notions below. 
\begin{definition}[Framing Space]
    Let $C$ be a framing space, $\ell: C \to B$ be a framing-to-belief mapping, 
    and $B = \{\ell(c) : c\in C \} \subseteq \Delta(\Omega)$ be the set of inducible prior beliefs.
    In a \emph{discrete} framing space, both $C$ and $B$ are discrete. In a \emph{continuous} framing space, $C$ is continuous and $B$ is assumed to be a convex subset of the belief simplex $\Delta(\Omega)$.
\end{definition}

\paragraph{\bf Sender-Receiver Interactions and the Equilibrium.} The interaction between sender and receiver is consistent with prior literature on persuasion, except that the sender's policy now additionally has a framing strategy.
The sender chooses a framing $c \in C$ prior to state observation ($c$ is thus state-independent from the sender's perspective); any state-dependent information is conveyed by the signal $s \sim \pi(\cdot|\omega)$.
This separation of state-independent framing and state-dependent signal is a key feature of our model. The formulation of information design with framing is 
described below.  

\begin{definition}[Information Design with Framing]\label{def:context_signaling}
    An \emph{information design policy with framing} is denoted by a tuple $\big(c \in C,~ \pi:\Omega\to\Delta(\Ss)\big)$.
    The sender first commits to a tuple $\big(c,~ \pi \big)$, and upon state realization $\omega \sim \mu_0$, the receiver observes the pair $(c \in C, s \sim \pi(\cdot| \omega))$ and updates his belief from $\mu_c = \ell(c)$ to the posterior $\mu_c(\omega \,|\, s) ~ = ~ \frac{\mu_c(\omega) \pi(s \,|\, \omega)}{\sum_{\omega' \in \Omega} \mu_c(\omega') \pi(s \,|\, \omega')}, \forall \omega \in \Omega$. The receiver then takes a best-response action that maximizes his expected utility under the updated belief $\mu_c(\cdot \,|\, s)$: 
    \begin{equation} \label{eq:receiver-argmax-definition}
        a^*_{c,\pi, s} \, \in ~  \argmax_{a\in \A} \sum_{\omega \in \Omega} \mu_c(\omega \,|\, s) v(a, \omega) ~ = ~ \argmax_{a\in \A} \sum_{\omega \in \Omega} \mu_c(\omega) \pi(s \,|\, \omega) v(a, \omega).
    \end{equation}
    Therefore, under an information design policy $(c,\pi)$, the sender's ex-ante (expected) utility is: 
    \begin{equation}\label{eq:sender_ex_ante}
        U(c, \pi) ~ := ~ \E_{\omega \sim \mu_0,\, s \sim \pi(\cdot|\omega)}\big[ u(a^*_{c, \pi, s}, \omega) \big] ~ = ~ \sum_{\omega \in \Omega} \mu_0(\omega) \sum_{s \in S} \pi(s \,|\, \omega) u(a^*_{c, \pi, s}, \omega).
    \end{equation}
\end{definition}

To map the above abstract model to our running example of advertising, this means the advertiser (sender) chooses a 
slogan (framing $c$) for a line of products, 
and based on each product's feature/quality (state $\omega$) chooses a discount 
(signal $s$).
The sender's goal is to choose an optimal  
framing and signaling strategy to maximize her expected utility assuming that the receiver will best-respond to the sender's strategy.  
This corresponds to a \emph{Stackelberg equilibrium}. 

We study two variants of the sender's information design problem with framing. 
The first is the optimization of framing under a given/fixed signaling scheme, coined \emph{framing-only optimization}. An instance of such a problem is denoted by $\I = (\mu_0, u, v, \pi)$ where $\pi$ is a given signaling scheme. The second is the \emph{joint optimization} of framing $c$ and signaling scheme $\pi$. 
Here, a problem instance is $\I = (\mu_0, u, v)$. 
Note that the third possible variant, optimizing the signaling scheme $\pi$ under a fixed framing (i.e., a fixed prior belief of the receiver) is essentially captured by 
previous work on Bayesian persuasion with heterogenous priors \citep{alonso_bayesian_2016} and thus not the focus of our work. 
In all cases, the sender is selecting a strategy to maximize her ex-ante utility with a best-responding receiver. 

\begin{definition}
\label{def:stackelberg-equilibrium}
The sender's \emph{Stackelberg equilibrium (i.e., optimal)} strategy is: 
\begin{itemize}
    \item 
    In a framing-only optimization problem $\I = (\mu_0, u, v, \pi)$, ~ $c^* = \argmax_{c \in C}{ U(c, \pi) }$; 
    \item In a joint optimization problem $\I = (\mu_0, u, v)$, ~ $(c^*, \pi^*) = \argmax_{c \in C,\, \pi: \Omega \to \Delta(\Ss)}{ U(c, \pi) }$.
\end{itemize}
\end{definition}

\section{Framing-Only Optimization}
\label{sec:context_only_optimization}
This section analyzes the conceptually simpler situation where we optimize the framing while fixing the signaling scheme.  Formally, given an information design instance $\I = (\mu_0, u, v, \pi)$, the sender aims to find the optimal framing $c^*$ in the Stackelberg equilibrium (Definition \ref{def:stackelberg-equilibrium}). We first specify an important distinction. The classic literature on Bayesian persuasion typically considers obedient signaling schemes, with signal space $\Ss$ equal to the action space $\A$, due to a revelation-principle-styled argument \citep{kamenica2011bayesian, dughmi_algorithmic_2016}. But in our setting, the signaling scheme $\pi$ is given exogenously and not within the designer's control, while the receiver's prior $\mu_c$ is optimized by the designer. As such, signals from $\pi$ cannot be freely interpreted as direct action recommendations, so we consider a general (non-direct) signaling space. 

We start by showing that randomizing over framings does not increase the expected utility of the sender under \emph{any} signaling scheme $\pi$. This means that it is without loss of generality to focus on deterministic framing both here and in the forthcoming section that studies jointly optimizing $\pi$ and $c$. All missing proofs from this section, including the one below, can be found in Appendix~\ref{appendix:context_only}.


\begin{proposition}\label{prop:no_randomization}
    For any instance $\I$ with a given signaling scheme $\pi$, the optimal sender utility can always be achieved by some deterministic framing $c^*$. 
\end{proposition}

As a simple corollary, for a discrete framing space, the optimal framing can be computed  in 
a polynomial time proportional to the size of the framing space, by enumerating the framing $c \in C$ and evaluating its sender utility $U(c, \pi)$. 
 \begin{corollary}
For a discrete framing space $C$, the optimal framing $c^*$ can be computed in $|C| \cdot \mathrm{poly}(|\Omega|, |\A|, |\Ss|)$ time. 
 \end{corollary}
 
 
 In practice (e.g., our running example of advertising), however, framing encompasses contextually relevant natural language expressions, hence the space $C$ is enormous. This gives rise to  a structurally more interesting question: are there algorithms that scale more gracefully with respect to the size of the framing space $C$? To study this problem, we turn to the relaxed problem that considers framing optimization in a continuous space.  


\subsection{Optimizing Framing in a Continuous Space}\label{subsec:cont_framing_only}

We turn to the continuous framing space with a convex belief space $B = \{\ell(c): c\in C\} \subseteq \Delta(\Omega)$ and consider directly optimizing in this space.
In the most optimistic case, the sender can use framing to induce \emph{any} receiver prior belief in the entire probability simplex: $B = \Delta(\Omega)$. This means that the framing space $C$ is no longer a parameter, and the instance size is completely determined by $|\Omega|$, $|\A|$, and $|\Ss|$. 
We will show that the continuous framing-only optimization problem is bi-directional equivalent to a classic problem in algorithmic game theory, \emph{Bayesian Stackelberg game}, but unfortunately the problem is NP-hard, even when $B = \Delta(\Omega)$. 



Formally, we directly optimize in the prior belief space $B = \Delta(\Omega)$. For any framing-induced receiver prior belief $\mu_c \in B$, let $U_\pi(\mu_c) = U(c, \pi)$ be the expected utility of sender, with fixed signaling scheme $\pi$. 
The framing-only optimization problem is to find an optimal $\mu^*_c \in B$ to solve: 
\begin{align} 
  \max_{\mu_c \in B=\Delta(\Omega)} & U_{\pi}(\mu_c) \, = \, \sum_{\omega \in \Omega}\sum_{s \in \Ss}{\mu_0(\omega)\pi(s|\omega)u(a^*_{\mu_{c}, s}, \omega)} \label{eq:framing-only-optimization-objective}\\
  \text{s.t.} \quad & a^*_{\mu_{c}, s} \in \argmax_{a 
\in \mathcal{A}}{\sum_{\omega' \in \Omega}{\mu_{c}(\omega')\pi(s|\omega')v(a, \omega')}}, ~~ \forall s\in \Ss. \nonumber 
\end{align}

We now connect 
the framing-only optimization problem 
to a classic problem in algorithmic game theory, \emph{Bayesian Stackelberg game} (Definition~\ref{definition:bsg}).
We then show in Proposition~\ref{prop:framing_to_bsg_reduction} that 
our framing-only optimization problem can be converted to finding the equilibrium of this game. 
\begin{definition}[Bayesian Stackelberg game]\label{definition:bsg}
    A \emph{Bayesian Stackelberg Game (BSG)} consists of a leader with action space $\mathcal{A}_{\ell}$ and a follower with action space $\mathcal{A}_{f}$ and private type $\theta$ drawn from a known distribution $P \in \Delta({\Theta})$, and type-dependent leader utility $u_{\ell}(a_{\ell}, a_{f}, \theta)$ and follower utility $u^{\theta}_{f}(a_{\ell}, a_f)$.
    The leader commits to a mixed strategy $x \in \Delta(\A_\ell)$, and the follower best responds. The leader's optimal (Stackelberg equilibrium) utility is given by:
    \begin{align*}
        \max_{x \in \Delta{(\mathcal{A}_{\ell})}} ~ & \sum_{\theta \in \Theta}{P(\theta)}\sum_{\al \in \A_\ell}{x(\al)u_{\ell}(a_{\ell}, a^*_f(\theta, x), \theta)} \\
        \mathrm{s.t.} \quad & a^*_f(\theta, x) \in \argmax_{a_f \in \mathcal{A}_f}{\sum_{a_{\ell} \in \A_\ell }x(a_{\ell})u^\theta_f(a_{\ell}, a_f)}, \quad \forall \theta \in \Theta. 
    \end{align*}
\end{definition}

\begin{proposition}\label{prop:framing_to_bsg_reduction}
The continuous framing-only optimization problem \eqref{eq:framing-only-optimization-objective} can be reduced to finding the Stackelberg equilibrium of a Bayesian Stackelberg game. 
\end{proposition}
\begin{proof}
We convert a continuous framing-only optimization problem to a BSG. 
Consider the BSG where the leader's action space is the space of states $\Omega$, 
the follower's type space is $\Ss$,
and 
the probability of type $s$ is $P(s) = \sum_{\omega \in \Omega} \mu_0(\omega) \pi(s|\omega)$.
The follower's utility function is $v^s(a, \omega) = \pi(s|\omega) v(a, \omega)$, and the leader's utility function is $ \tilde u(a, s) = \sum_{\omega \in \Omega} \frac{\mu_0(\omega)\pi(s|\omega)}{P(s)}u(a, \omega)$. Note that the leader's utility $\tilde u(a, s)$ depends on the follower's action $a$ and type $s$, but not the leader's action $\omega$. 
By definition, the leader's optimization problem in the BSG is
\begin{align}
  \max_{x \in \Delta(\Omega)} ~ & U_{P}(x) \, = \, \sum_{s \in \Ss} P(s) \tilde u(a^*(s, x), s)  \\
  \text{s.t.} \quad & a^*(s, x) \in \argmax_{a 
\in \mathcal{A}}{\sum_{\omega' \in \Omega}{x(\omega') \tilde v^s(a, \omega')}}, ~~ \forall s\in \Ss, \nonumber 
\end{align}
namely, 
\begin{align}
  \label{eq:BSG-formulation-objective}
  \max_{x \in \Delta(\Omega)} ~~ & U_{P}(x) \, = \, \sum_{s \in \Ss} P(s) \sum_{\omega \in \Omega} \frac{\mu_0(\omega)\pi(s|\omega)}{P(s)}u(a^*(s, x), \omega)  \\
  \text{s.t.} \quad & a^*(s, x) \in \argmax_{a 
\in \mathcal{A}}{\sum_{\omega' \in \Omega}{x(\omega') \pi(s|\omega')v(a, \omega')}}, ~~ \forall s\in \Ss. \nonumber 
\end{align}
We note that the optimization problems \eqref{eq:BSG-formulation-objective} and \eqref{eq:framing-only-optimization-objective} are equivalent, where $x$ corresponds to $\mu_c$ and $a^*(s, x)$ corresponds to $a^*_{\mu_c, s}$. Therefore, we conclude that the continuous framing-only optimization problem can be reduced to a Bayesian Stackelberg game.
\end{proof}

The reduction from framing-only optimization to Bayesian Stackelberg games suggests that one approach to finding the optimal framing is to use existing computational methods for BSG (e.g., \citep{paruchuri2008playing}). 
However, \citet{conitzer2006computing} show that a family of BSGs are computationally intractable (NP-hard).
Our Theorem \ref{thrm:np_hardness} proves that any BSG in that family can be converted to a continuous framing-only optimization problem. 
Combining with Proposition~\ref{prop:framing_to_bsg_reduction}, our analyses show that framing-only optimization is (bi-directional) equivalent to a subset of BSGs that are known to be computationally hard.
Hence, the framing-only optimization problem is NP-hard and not easily solvable even with access to a perfect mapping $\ell$ from framing to belief. 



\begin{theorem}[NP-hardness] \label{thrm:np_hardness}
    For the continuous framing-only problem $\I = (\mu_0, u, v, \pi)$, 
    there is no additive Fully Polynomial-Time Approximation Scheme (FPTAS) for computing the optimal framing-induced prior belief $\mu_c^* \in B$ in \eqref{eq:framing-only-optimization-objective}, unless P = NP. This hardness holds even when $B = \Delta(\Omega)$. 
\end{theorem}





\begin{proof}[Proof Sketch]
    The full proof is technical and given in Appendix \ref{proof:np-hardness}; we sketch the high-level intuition here. 
    \citet{conitzer2006computing} show that a family of BSGs where the followers have binary actions is NP-hard. 
    We show that any such BSG can be cast into a framing-only optimization problem $\I$ with $|\Omega| = |\A_{\ell}| + |\Theta| + 1$ states, $|\A| = |\Theta||\A_f| + 2$ actions, $|\Ss| =  |\Theta|$ signals, and a continuous framing space $B = \Delta(\Omega)$. Specifically, we create a state for each leader action, $\omega_{\al}$, and each follower type $\omega_{\theta}$. We create receiver actions for each binary action a follower of a type $\theta$ can take -- i.e. $a^\theta_i$ for all $\theta$. When the receiver sees a signal $s_{\theta}$ (which is proxying type $\theta$ in BSG), we want them to only consider actions $\atheta_0, \atheta_1$, which should directly correspond to follower $\theta$'s utility in taking action 0 or 1 in the BSG. Since the receiver's utility in the information design problem does not explicitly depend on type $\theta$, we use the states $\omega_{\theta}$ to achieve this effect. The receiver is heavily penalized for taking an action $a^{\theta'}_{*}$ at state $\omega_{\theta}$. We carefully construct the sender's utilities and add additional dummy states and actions to ensure that (1) on the optimal prior $\mu_c^*$, the sender places sufficient weight on states corresponding to $\omega_{\theta}$ to ensure that the receiver takes this type-consistent action, and (2) the information design objective captures the type-dependent Bayesian Stackelberg objective. The inapproximability stems from a more careful analysis of the original result of \cite{conitzer2006computing}.
\end{proof}

\subsection{The Effect of Framing on  Sender Utility}\label{sec:effect-framing-only}
Next, we analyze how much the sender can improve her utility via manipulating framing with a given signaling scheme $\pi$.  To build intuition for this section's result, we start with an example to demonstrate how much a belief shift induced by framing can lead to a meaningful improvement in the sender's utility -- is a substantial belief change necessary?

The following variant of the well-known prosecutor-judge example in \citet{kamenica2011bayesian} shows that the answer is \emph{No}, meaning that a small change in framing can significantly change the sender's utility.
Specifically, consider a Bayesian persuasion setting where a defendant may be either \emph{innocent} or \emph{guilty} (two possible states of the world), and the judge (receiver) decides whether to \emph{acquit} or \emph{convict} the defendant, receiving utility $1$ for the just action and 0 otherwise. The prosecutor (sender) observes the defendant's true state and can signal accordingly to maximize their utility, which is $1$ for a conviction regardless of the state. Now suppose the prosecutor, when innocent, always recommends acquittal, and when guilty, recommends either action with probability $0.5$. If the prosecutor's prior belief over the states is $[0.67, 0.33]$ and the judge shares this belief (as is the case in Bayesian Persuasion), then the prosecutor achieves utility $0$. If, however, the prosecutor can use framing (style of argument/language) to slightly alter the judge's prior belief to $[\tfrac{2}{3}, \tfrac{1}{3}]$, then the prosecutor's utility under the same signaling scheme jumps to $0.66$! This is because the sender's utility as a function of the receiver's belief 
is not continuous for the given signaling scheme. This discontinuity highlights the power of leveraging framings: for a fixed signaling scheme, slightly altering the receiver's belief by framing can have a major impact on the receiver's action and hence the sender's utility. 

To formalize this, let $U_\pi(\mu)$ denote the sender's expected utility \eqref{eq:framing-only-optimization-objective} under fixed signaling scheme $\pi$ when the receiver's prior belief is $\mu$. In the prosecutor-judge instance, this function is discontinuous at $\mu = [\tfrac{2}{3}, \tfrac{1}{3}]$. We show below that such discontinuities occur in general instances almost surely for a large class of signaling schemes, specifically, schemes in which some signal $s$ is sent with positive probability at every state. In other words, schemes not satisfying this condition must be ``very revealing'': upon observing any signal $s$, the receiver can rule out at least some state(s) with full confidence. The proof of this result is in Appendix~\ref{proof:fixed-scheme-discontinuous}. 

\begin{proposition}[Discontinuous sender utility]
\label{prop:fixed-scheme-discontinuous}
Consider any signaling scheme $\pi$ in which there exists a signal $s_0 \in \Ss$ such that $\pi(s_0|\omega) > 0$ for every state  $\omega \in \Omega$. Then the sender's expected utility $U_\pi(\mu)$ as a function of the receiver's prior belief $\mu$ is \emph{generally discontinuous} in the following sense: for any $\mu_0$, for $u$ and $v$ sampled from any continuous distribution over utility functions, $U_\pi(\mu)$ is discontinuous in $\mu \in \Delta(\Omega)$ with probability 1.
This holds even if 
$u(a, \omega)$ is independent of $\omega$.
\end{proposition}
This discontinuity also implies that the sender's utility is highly sensitive to errors in
a simulated framing-to-belief mapping $\ell$. 
Suppose using an imperfect mapping $\ell_\varepsilon$ with error $\eps$, we find an optimal framing $\hat{c}$ 
that happens to induce
a belief $\hat{\mu} = \ell_\eps(\hat c)$ such that $U_{\pi}(\cdot)$ was discontinuous at $\hat{\mu}$. Then if this framing is deployed,
even though the true induced prior $\mu^* = \ell(\hat c)$ is within distance $\varepsilon$ to $\hat{\mu}$ (i.e., $|\mu^* - \hat \mu| \le \eps$), the discontinuity means that the realized utility can be arbitrarily far from the utility achieved under the imperfect mapping, regardless of how small $\eps$ is.
In particular, there does not exist a scalar $\lambda$ such that $|U_{\pi}(\hat{\mu}) - U_{\pi}(\mu^*)| \leq \lambda \varepsilon$. Indeed, this discontinuity is the fundamental reason for the hardness result shown in Theorem \ref{thrm:np_hardness}. 

\section{Framing-Signaling Joint Optimization}\label{sec:joint}

We now consider the joint design of framing and signaling scheme: 
given an instance $\I = (\mu_0, u, v)$, the sender selects a framing $c \in C$ and a signaling scheme $\pi$ to maximize her utility $U(c, \pi)$. This is a \emph{strict} generalization of the standard Bayesian persuasion model, which is the special case where $C$ is a singleton set hence only the signaling scheme $\pi$ is a variable. As shown by Proposition \ref{prop:no_randomization}, 
there is no benefit in randomizing over framings, so we focus on deterministic framing. 

The ability to design $\pi$ offers the sender more freedom as compared to framing-only optimization. Indeed, a key challenge in the restricted framing-only case was the inability to interpret signals as action recommendations. 
But for the design of joint strategy, it is without loss of generality to focus on 
obedient recommendation schemes, as we show below (the proof is in Appendix~\ref{appendix:joint}). 

\begin{observation}\label{ob:revelation_principle}
To compute the optimal joint framing-signaling strategy $(c^*, \pi^*)$, 
it suffices to consider signaling scheme $\pi$ with a direct signal space, i.e., $\Ss = \A$, in which each signal $a \in \A$ makes an obedient action recommendation to the receiver.
\end{observation}

Observation \ref{ob:revelation_principle} allows us to restrict attention to the design of direct signaling schemes without loss of generality.
This succinct representation gives rise to the following corollary about polynomial-time solvability under discrete framing space. The algorithm simply enumerates all possible $c\in C$, and for each $c$ computes the corresponding optimal signaling scheme via the standard Bayesian persuasion linear program \cite{dughmi_algorithmic_2016}. 
\begin{corollary}
For a discrete framing space $C$, the optimal joint strategy $(c^*, \pi^*)$ can be computed in $|C|\cdot \mathrm{poly}( |\A|,  |\Omega|)$ time.     
\end{corollary} 


Similar to Section \ref{sec:context_only_optimization}, next we answer two key questions: (1) the effect of framing on the sender's utility and the sensitivity to errors in framing-to-belief simulation, and (2) the efficient computability of the optimal joint strategy with perfect simulation. 
 Our results highlight key differences between joint design and framing-only design.  The joint design problem turns out to be more tractable. 

\subsection{The Effect of Framing on Sender Utility}
To compare with Proposition \ref{prop:fixed-scheme-discontinuous}, we start our analysis by analyzing the continuity property of the optimal sender utility as a function of the framing-induced receiver prior belief $\mu$. It is not difficult to see that, in the optimal joint design, given any framing-induced receiver prior belief $\mu$, the sender's corresponding optimal signaling scheme is the one that optimally accompanies this receiver prior belief. We denote this signaling scheme as $\pi^*_{\mu}$ and the resulting sender utility $U^*(\mu)$ as a function of the framing-induced receiver prior belief $\mu$. Observation \ref{ob:revelation_principle} shows that $U^*(\mu)$ can be efficiently computed for any $\mu \in \Delta(\Omega)$ using the following linear program, with the constraints referred to as the \emph{obedience} constraints as is standard in the literature \cite{bergemann2019information}:
\begin{align}
    U^*(\mu) ~ := ~ \max_{\pi : \Omega \rightarrow \Delta(\A)} ~ & \sum_{\omega \in \Omega}\sum_{a \in \A}{\mu_0(\omega)\pi(a|\omega)u(a, \omega)} \label{equation:joint_lp}\\
    \text{s.t. } \quad & \sum_{\omega \in \Omega}{\mu(\omega)\pi(a|\omega)\big[v(a, \omega) - v(a', \omega)\big] \geq 0}, \quad \forall a, a' \in \A \times \A. \label{eq:joint-lp-constraint}
\end{align}
The $U^*(\mu)$ above tracks the sender's utility as a function of the receiver's prior belief $\mu$ under an optimal signaling scheme. It is worthwhile to compare $U^*(\mu)$ to the earlier function $U_{\pi}(\mu)$ introduced in Section \ref{sec:context_only_optimization} to capture the same quantity but with a fixed signaling scheme $\pi$. In contrast to the discontinuity of $U_{\pi}(\mu)$ shown in Proposition \ref{prop:fixed-scheme-discontinuous}, our following Theorem \ref{thrm:joint_continuous} (proof deferred to Appendix~\ref{proof:joint-continuous}) shows that once $\pi$ becomes part of the strategy space, the sender's optimal utility becomes continuous in the receiver's prior belief $\mu$. 

\begin{theorem}[Continuous sender utility]
\label{thrm:joint_continuous}
The sender's utility $U^*(\mu)$, defined in \eqref{equation:joint_lp}, is a locally Lipschitz continuous function of the induced receiver prior belief $\mu$ in the interior of the simplex
$\Delta(\Omega)$. 
\end{theorem}
\begin{proof}[Proof Sketch]
The high-level idea is a sensitivity analysis for the linear program \eqref{equation:joint_lp}\eqref{eq:joint-lp-constraint}, whose variable is $\pi$ and parameter is $\mu$. We want to show that the optimal objective $U^*(\mu)$ 
cannot change a lot when the parameter $\mu$ is slightly perturbed. In particular, let $\pi^*_\mu$ be an optimal solution of the linear program when the parameter is $\mu$. We modify $\pi^*_\mu$ slightly to be another solution $\tilde \pi$ that satisfies the obedience constraint \eqref{eq:joint-lp-constraint} simultaneously for all parameters $\mu'$ that are close to $\mu$ (in the sense of $\|\mu' - \mu \|_1 \le \eps)$. Such modification is possible due to the conditions that (1) every action $a\in\A$ of the receiver is strictly inducible by some belief; (2) $\mu(\omega) > 0$ for every $\omega \in \Omega$.  Since the modification is small, the utility of $\tilde \pi$ is only slightly worse than the utility of $\pi^*_\mu$, which is $U^*(\mu)$.  So, the optimal objective $U^*(\mu')$ with parameter $\mu'$ cannot be too much worse than $U^*(\mu)$, which establishes the continuity of the $U^*(\cdot)$ function. See details in Appendix~\ref{proof:joint-continuous}. 
\end{proof}

Theorem \ref{thrm:joint_continuous} has useful implications on understanding the sender's utility loss due to the error in a simulated belief-to-framing mapping $\ell_{\varepsilon}$.
To illustrate this, it is helpful to relax the obedience constraint. Formally, we say an action recommendaiton $a \in \A$ is \emph{$\varepsilon$-obedient} for a receiver with prior belief $\mu$ if for every action $a' \in \A$, $\sum_{\omega}\mu (\omega)\pi(a|\omega)\big[v(a, \omega) - v(a', \omega) \big] \geq -\eps$.\footnote{As per convention, $\varepsilon$-obedience means the obedience constraints may be violated by at most $\varepsilon$. The notation $(-\varepsilon)$-obedience thus means the obedience constraints are strictly satisfied with an $\eps$ margin.}
Now let $(\mu^*, \pi^*_{\mu^*})$ denote an optimal joint strategy, where prior belief $\mu^* = \argmax_{\mu \in B} U^*(\mu)$, with framing $c^*$ inducing $\ell(c^*) = \mu^*$ under the true belief-to-framing mapping $\ell$. 
%
%
Choosing $c^*$ under the imperfect mapping $\ell_{\varepsilon}$ will result in a perceived belief $\ell_{\varepsilon}(c^*)$ within the ball $B_{\varepsilon}(\mu^*) = \{ \mu \in \Delta(\Omega) : \| \mu - \mu^* \| \le \eps \}$, and due to Theorem~\ref{thrm:joint_continuous},
the perceived utility will be within $\lambda_1 \varepsilon$ of the optimal, where $\lambda_1$ is the local Lipschitz constant of $U^*$, hence $| U^*( \ell_\eps(c^*)) - U^*(\mu^*) | \le \lambda_1 \eps$.
If we find the optimal joint strategy $(\hat \mu, \pi^*_{\hat \mu})$ under the imperfect mapping $\ell_\eps$, with $\hat \mu = \argmax_{\mu \in \hat B} U^*(\mu)$ where $\hat B = \ell_\eps(C)$, and find the corresponding framing $\hat c$ inducing $\ell_\eps(\hat \mu) = \hat \mu$, then we have: $U^*(\hat \mu) \ge U^*(\ell_\eps(c^*)) \ge U^*(\mu^*) - \lambda_1 \eps$.


\vspace{0.5em}

\noindent
\begin{minipage}{0.47\textwidth}
  \includegraphics[width=\linewidth]{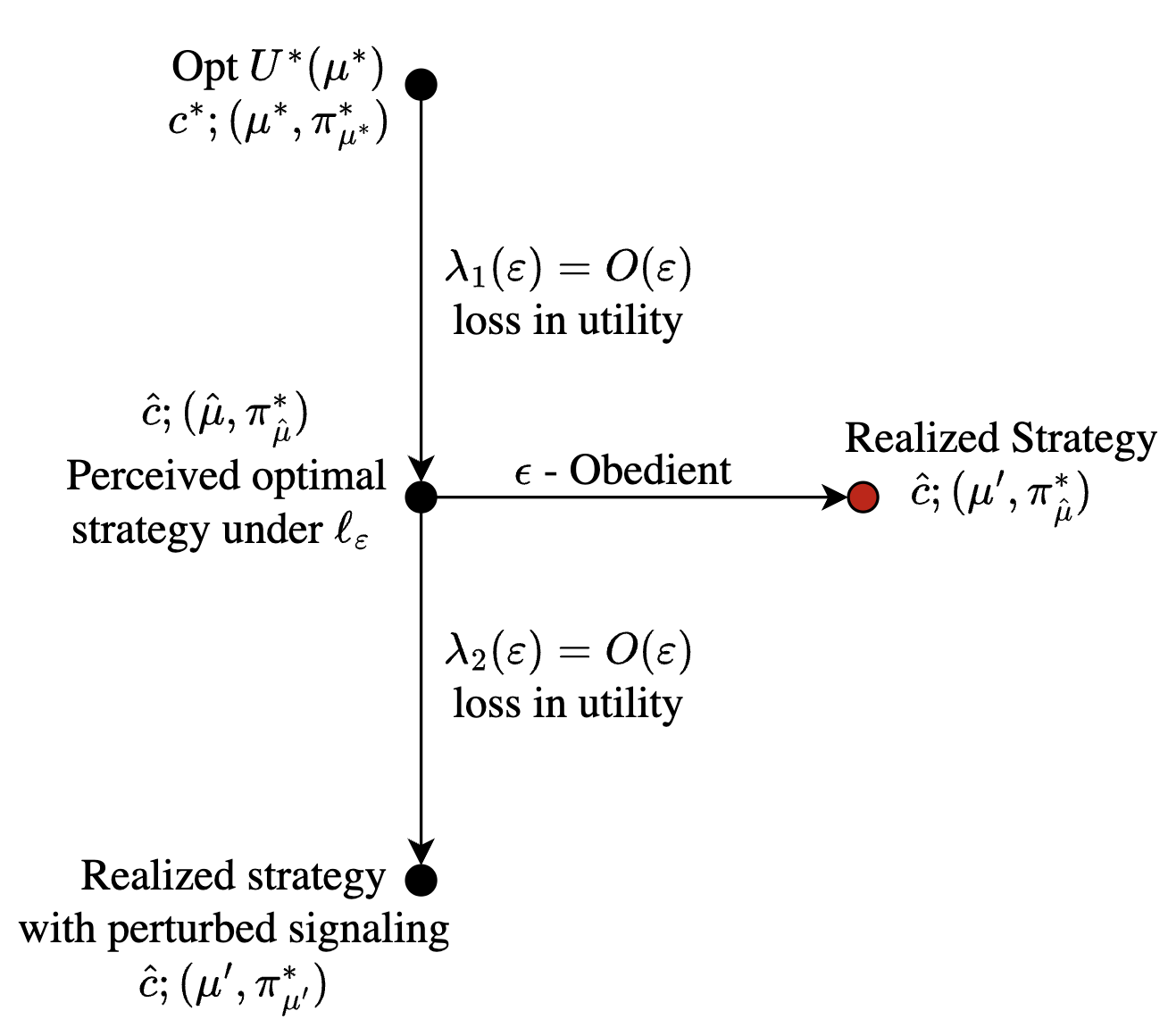}
\end{minipage}%
\hfill
\begin{minipage}{0.51\textwidth} 
When using framing $\hat{c}$ in practice, however, the actual belief induced by $\hat c$ is $\mu' = \ell(\hat c) \ne \hat \mu$ where $\| \mu' - \hat{\mu}\| \leq \varepsilon$. Since the signaling $\pi^*_{\hat{\mu}}$ is unchanged, when the prior belief is shifted from $\hat{\mu}$ to $\mu'$, it is evident that the obedience constraints, which were exactly satisfied for $\hat{\mu}$, are violated by at most $\varepsilon$ for $\mu'$ -- in other words, the deployed scheme is $\varepsilon$-obedient. 
If the sender was aware of this error and wished to be conservative, she could slightly modify $\pi^*_{\hat{\mu}}$ to make it exactly obedient for $\mu'$: this is feasible because $\mu'$ is close to $\hat{\mu}$, but will cause an additional $\lambda_2 \eps$ loss to the sender's utility. 
The total loss is $\lambda_1\eps + \lambda_2 \eps = O(\eps)$. 
We illustrate this visually on the left and state it formally below:
\end{minipage}

 \begin{corollary}\label{cor:joint_error_bounded}
    The realized utility in facing an $\varepsilon$-obedient receiver under a joint optimal strategy based on an imperfect mapping $\ell_{\varepsilon}$ is at most $ \lambda_1 \epsilon = O(\varepsilon)$ away from the optimal utility $U^*(\mu^*)$ for an exactly obedient receiver. In slightly modifying the signaling scheme of this strategy, 
    the strategy can be exactly obedient and at most $\lambda_1 \varepsilon + \lambda_2 \varepsilon = O(\eps)$ away from the optimal utility $U^*(\mu^*)$. 
\end{corollary}

 
We will thus allow the design of $\varepsilon$-obedient signaling schemes throughout this section. This is primarily for mathematical convenience since under mild regularity assumptions, the loss in the obedience constraint can be transferred into the loss in the objective with the same order. The main idea is to compute strictly obedient schemes by adding an $O(\varepsilon)$ amount of margin to obedience constraints, hence this scheme is still obedient
even when the receiver's prior belief is off by an $\varepsilon$ amount. 
What remains is then to bound the sender's utility loss due to imposing a stricter obedience constraint, which can be bounded as $O(\epsilon)$ 
(see, e.g., \cite{zu_learning_2021, lin_information_2025}).  

\subsection{Optimizing the Joint Strategy in a Continuous Framing Space}
Under discrete framing space, the optimal joint strategy can be computed in time linear in $|C|$, the size of framing space. This space in reality can be enormous, hence it is more insightful to consider continuous yet naturally structured framing space. Specifically, we consider the case where the framing-induced prior belief space $B = \{ \ell(c): c\in C\} \subseteq \Delta(\Omega)$ is convex.
As in Section~\ref{subsec:cont_framing_only}, we consider directly optimizing in the belief space $B$.
This leads to the following optimization problem with a linear objective subject to \emph{bi-linear} constraints:
\begin{align}\label{eq:IC_constraints_joint}
    \max_{\mu \in B}{U^*(\mu)} ~ =~ \max_{\mu \in B,\, \pi: \Omega\to\Delta(\A)}~ & \sum_{\omega \in \Omega}\sum_{a \in \A}{\mu_0(\omega)\pi(a|\omega)u(a, \omega)} \\ \nonumber
    \text{s.t. } \qquad & \sum_{\omega \in \Omega}\mu (\omega)\pi(a|\omega)\big[v(a, \omega) - v(a', \omega)\big] \geq 0, \quad \forall \, a, a' \in \A \times \A. 
\end{align}
Bi-linear optimization problems are well-known to be challenging to solve in general. In the following example, we illustrate this challenge by demonstrating the non-convexity and non-concavity of $U^*(\mu)$ even in simple instances.  

\vspace{0.5em}

\noindent
\begin{minipage}{0.55\textwidth}
    \begin{example}[Non-convexity and non-concavity of $U^*(\mu)$]
    \label{example:joint}
Consider an instance with $2$ states $\Omega=\{0, 1\}$ and 3 actions $\{0, 1, 2\}$, with the following utility matrices for the sender and the receiver (rows are actions, columns are states): 
\begin{align*}
    u = \begin{bmatrix}
        0 & 1 \\
        1 & 0 \\
        0.2 & 0.2 
    \end{bmatrix}, \quad \quad 
    v = \begin{bmatrix}
        0.65 & 0.15 \\
        0.60 & 0.30 \\
        0.10 & 0.50 
    \end{bmatrix}. 
\end{align*}
The sender has prior $\mu_0 = (\tfrac{1}{3}, \tfrac{2}{3})$ for the two states. We use the probability of state $0$ to denote the receiver's belief $\mu \in [0, 1]$. The sender's optimal utility function $U^*(\mu)$ is plotted to the right. It is continuous but not convex, concave, or quasi-concave. 
\end{example}
\end{minipage}
\hfill
\begin{minipage}{0.42\textwidth}
    \centering
    \includegraphics[width=\textwidth]{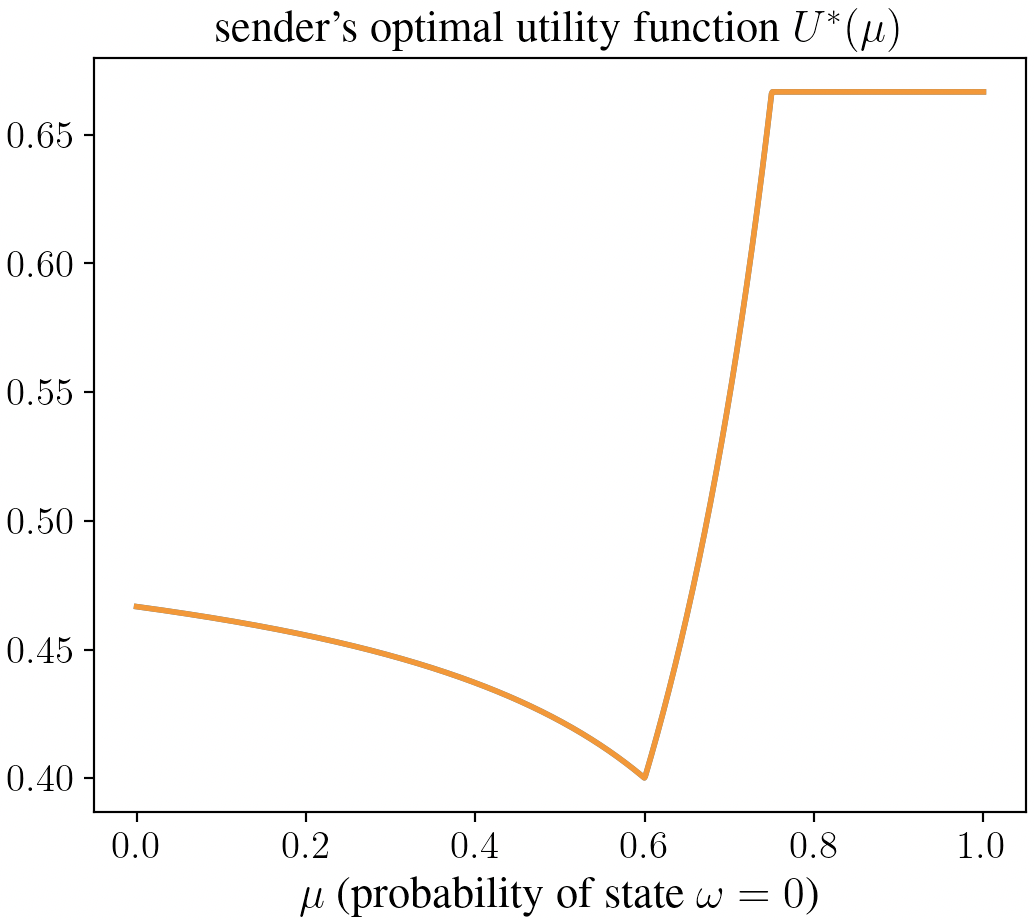} 
\end{minipage}

\vspace{0.5em}

The example above hints at potential barriers for the optimization problem.
Indeed, recall that the framing-only problem in Section \ref{sec:context_only_optimization} is NP-hard even when there is no constraint on the induced belief space $B$, i.e., $B = \Delta(\Omega)$ is the entire simplex.
But interestingly, we obtain a positive result for unconstrained belief space in the joint design case: there is 
an algorithm with $\frac{1}{|\Omega|} $ loss in objective and $\varepsilon$ loss in obedience constraint with $\mathrm{poly}\left({|\Omega|, |A|, 1/\varepsilon}\right)$ running time.
Further, when the sender's utility is state-independent and only depends on the receiver's action (e.g., the prosecutor-judge example), the exactly optimal policy can be computed in polynomial time. 

\begin{theorem}
[Positive results for unconstrained belief space]
\label{thm:joint-unconstrained}
For the framing-signaling joint optimization problem 
with unconstrained belief space $B = \Delta(\Omega)$, the following hold:
\begin{itemize}
     \item A $\big(1 - \tfrac{1}{|\Omega|}\big)$ multiplicative-approximation of the optimal joint strategy utility can be computed in poly-time 
     under $\varepsilon$-obedience constraint; 
     \item If the sender's utility is state-independent, i.e., $\forall a, \forall (\omega, \omega') ,\, u(\omega, a) = u(\omega', a)$, then the exactly optimal strategy can be computed in poly-time. 
 \end{itemize}
\end{theorem}
\begin{proof}
    See Appendix \ref{proof:joint-unconstrained}. 
\end{proof}

Theorem \ref{thm:joint-unconstrained} studies unconstrained belief space. Next, we turn to the constrained (general) belief space and present a 
Quasi-Polynomial Time Approximation Scheme (QPTAS) in Theorem \ref{thm:qptas}, which computes an $\eps$-obedient joint strategy that is at least as good as the optimal obedient strategy in $\mathrm{poly}(|\Omega|^{\frac{2\log |\A|}{\varepsilon^2}})$ time. 
While the hardness of the exact optimization problem is an interesting open question, we view the existence of a QPTAS 
as evidence 
that the joint optimization problem is more tractable than the framing-only variant in Section \ref{sec:context_only_optimization}, which was shown to be reducible from the \emph{Independent-Set} problem that has no known QPTAS. 

\begin{theorem}
[QPTAS for constrained belief space]
\label{thm:qptas}
For any instance $\I = (\mu_0, u, v)$ with convex belief space $B \subseteq \Delta(\Omega)$, for any $\varepsilon > 0$, there exists a $\mathrm{poly}(|\Omega|^{\frac{4\log |\A|}{\varepsilon^2}})$-time algorithm that computes an $\varepsilon$-obedient joint strategy
with utility at least as the optimal joint strategy under exact obedience.
\end{theorem}

\begin{proof}
    Let $\mu^* \in B$ and $\pi^* : \Omega \to \Delta(A)$ be the optimal induced prior belief and signaling scheme for this instance under exact obedience constraints. As $\mu^*$ is a distribution over $\Omega$, we can draw $n$ samples from $\mu^*$, and the resulting empirical distribution, denoted $\muhat$, is an $n$-uniform distribution -- i.e., each entry of $\muhat$ is a multiple of $\tfrac{1}{n}$.
    Because
    \[ \E\Big[\sum_{\omega \in \Omega}\muhat(\omega)\pi^*(a | \omega)\big[ v(a, \omega) - v(a', \omega) \big] \Big] = \sum_{\omega \in \Omega}\mu^*(\omega)\pi^*(a | \omega)\big[v(a, \omega) - v(a', \omega) \big] \ge 0,\]
    by Hoeffding's inequality, we have 
    \begin{equation*}
        \forall a, a' \in \A \times \A, \quad \Pr\Big[ \sum_{\omega \in \Omega}{\muhat(\omega)\pi^*(a | \omega)\big[v(a, \omega) - v(a', \omega)\big]} < -\eps \Big] \leq \exp\left(-n\eps^2/2 \right). 
    \end{equation*}
    Taking a union bound over all $|\A|^2$ pairs $(a, a')$, we have $\sum_{\omega}{\muhat(\omega)\pi^*(a | \omega)\big[v(a, \omega) - v(a', \omega)\big]} \ge -\eps$ satisfied for all pairs of $(a, a')$ with probability at least $1 - |\A|^2\exp\left(-n\eps^2/2\right)$, hence $\hat \mu$ in conjunction with $\pi^*$ satisfies $\eps$-obedience constraints. Pick $n = \frac{4 \log |A|}{\eps^2}$. The probability $1 - |\A|^2\exp\left(-n\eps^2/2\right)$ will be positive.  This means that there must exist an $n$-uniform distribution $\hat \mu$ satisfying $\eps$-obedience in conjunction with $\pi^*$. Note that the sender's utility under the $(\hat \mu, \pi^*)$ strategy is the same as the $(\mu^*, \pi^*)$ strategy because the sender's utility depends on $\pi^*$ but not the receiver's prior belief. 

    Now consider the following algorithm: enumerate over all $n$-uniform distributions in $B$, and for each, solve the optimal signaling linear program \eqref{equation:joint_lp}-\eqref{eq:joint-lp-constraint} but with a relaxed $\varepsilon$-obedience constraint.\footnote{We assume that $B$ contains at least one $n$-uniform distribution, a mild assumption given a rich $B$. If not, then we project the closest $n$-uniform distribution to $B$, which, because the set of $n$-uniform distributions is an $O(\frac{|\Omega|}{n})$-cover of $\Delta(\Omega)$ under $\ell_1$ distance, only causes $O(\frac{|\Omega|}{n}) = O(\frac{|\Omega| \eps^2}{\log|A|})$ loss to the sender's utility.}
    Return the solution with the best sender utility.  This solution must be weakly better than the $\eps$-obedience solution $(\hat \mu, \pi^*)$ mentioned above, which is therefore weakly better than the optimal solution $(\mu^*, \pi^*)$.

    We then consider the runtime of the algorithm. 
    The runtime depends on the number of $n$-uniform distributions and the time to check whether each distribution is included within the inducible belief set $B$. Since $B$ is convex, checking this inclusion can be done in poly-time. As for the number of $n$-uniform distributions, since the probability of each state $\omega$ can take on $n$ possible values $\{0, \tfrac{1}{n}, \tfrac{2}{n}, \dots, 1\}$ and the sum must equal $1$, it is equivalent to placing $n$ elements into $|\Omega|$ distinct buckets. Thus, the number of possible distributions is at most $\binom{n + |\Omega| - 1}{|\Omega| - 1}$.
    For a fixed $|\Omega|$ and $n$ growing large, this quantity is a polynomial in $n$ of degree $|\Omega|-1$, which is clearly upper bounded by $O(|\Omega|^n)$. 
    For a fixed $n$ and as $|\Omega|$ grows large, we have $\binom{n + |\Omega| - 1}{|\Omega| - 1} = \binom{n + |\Omega| - 1}{n} \le \frac{ (n + |\Omega| - 1)^n}{n!} = O(|\Omega|^n)$.
    Thus, the runtime of this algorithm is bounded by $\mathrm{poly} \cdot O(|\Omega|^n) = \mathrm{poly}(|\Omega|^{\frac{2\log |\A|}{\varepsilon^2}})$. 
\end{proof}

\section{Empirical Study with Large Language Models}\label{sec:experiment}
Our theoretical result illustrates that, while optimizing the framing-induced prior belief $\mu_c$ (let alone the language framing $c$) for a fixed signaling scheme $\pi$ is computationally difficult, jointly optimizing both is more tractable.
So, we conduct an empirical study on the joint optimization problem.
We will co-design the framing-signaling pair $(c, \pi)$, as opposed to directly optimizing the belief-signaling pair $(\mu_c, \pi)$ as in the theoretical section.
Since $c$ belongs to a rich natural language space, our investigation here 
makes use of large language models (LLMs). 

There are two possible approaches to 
optimizing the framing-signaling pair $(c, \pi)$. 
One approach is to first leverage Theorems~\ref{thm:joint-unconstrained} or \ref{thm:qptas} to find the optimal $(\mu_c, \pi)$ pair in quasi-polynomial time, 
then appeal to LLMs to approximate the function $\ell^{-1}$ that maps the belief $\mu_c$ back to a framing $c$ that would induce that belief. While our theoretical results do much of the heavy-lifting for this approach, it suffers from one crucial drawback: it requires \emph{a priori} knowledge of the inducible belief set $B$, which in many cases may 
be unavailable. 
This motivates a second approach inspired by \emph{iterative improvement/hill-climbing} in prompt optimization: Leverage LLMs to (1) directly search over the language space of framings and (2) approximate the belief mapping $\ell$. For any candidate framing $c$ with LLM approximated induced belief $\mu_c$, the corresponding optimal signaling scheme can be computed using the linear program in Equation~\eqref{eq:joint-lp-constraint}. This is used as feedback to guide the LLM framing search. We elaborate on this approach below; as an aside, this approach can also be used to search over framings under a fixed signaling scheme, i.e., the framing-only problem variant.

\subsection{Methodology}\label{sec:methodology}
We present our proposed approach for optimizing in the framing space in Figure \ref{fig:optimization_flow}. LLMs search the framing space and estimate the prior belief of a receiver induced by a framing.
These two LLMs modules are complemented with an additional LLM that verifies the soundness/correctness of any framing, and analytical solvers that compute the sender's utility under the optimal signaling scheme for a given quantitative receiver belief. We now discuss each of these roles:

\begin{figure}[t]
    \centering
    \includegraphics[width=0.95\linewidth]{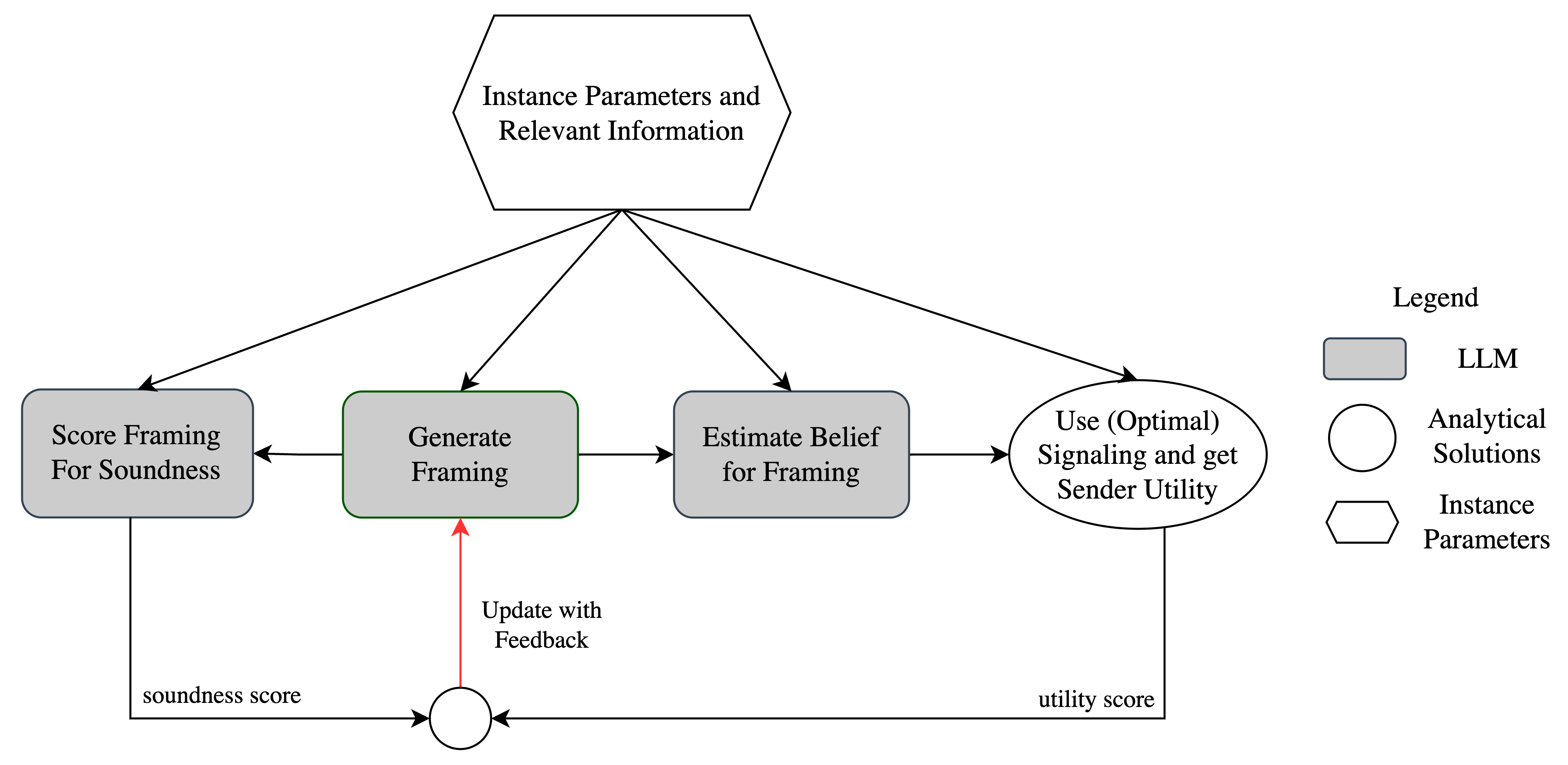}
    \caption{Diagram of our proposed framework for optimizing framing and signaling. It includes LLMs searching the framing space, verifying the correctness of framing, and generating framing-induced beliefs. It also includes poly-time analytical solvers to compute optimal signaling scheme for a given belief.}
    \label{fig:optimization_flow}
    \vspace{-1em}
\end{figure}

\textbf{Instance information:} Information relevant for optimal framing and signaling is not just the quantitative parameters like $(u, v)$ but also qualitative descriptions about the setting, the receiver and sender, and any related meta information. This is because the belief formation process induced by framing 
is not explicit or mathematical. 
Rather, it captures how a string describing some aspect of the instance will be perceived by a given receiver, factoring in social, environmental, and personal factors. It is thus important to provide these additional contexts. 

\textbf{Estimating the framing-to-belief mapping $\ell$:}
We use an LLM to estimate how a framing would influence a given receiver/decision maker's belief. In cases where the decision-making has been delegated to LLMs, an increasingly common scenario~\citep{fish2023generative}, we can simply query that LLM and the approximation is essentially exact. When the decision-maker is human, a nascent line of economics research argues that LLMs can approximate the human decision-maker in various settings~\citep{horton2023large, manning2024automated}, including as a statistical proxy to model human beliefs \citep{namikoshi2024using}. This involves endowing the LLM with information about the human it is modeling --  demographics, preferences, and backgrounds. What we precisely require from the LLM is that it outputs a quantitative estimation of the receiver's belief. This can be gathered by either: asking the LLM to generate one of $|\Omega|$ tokens corresponding to each state and recording the log probabilities, or directly asking it to return numerical probabilities. \citet{cruz2024evaluating} show that on distributional (non-factual) questions, the first approach leads to uncalibrated answers (the log probabilities are far from the true distribution), while direct elicitation in a chat-style prompt has better outcomes.
We thus use the direct elicitation approach and comment on our empirical observations in Section~\ref{subsubsec:belief_estimation}. Lastly, Theorem~\ref{thrm:joint_continuous} and Corollary~\ref{cor:joint_error_bounded} imply that any errors in this belief approximation will only have 
small influences on the sender's utility, which is practically helpful. 

\textbf{Validation of framing:} Using LLMs to generate framing risks hallucinations. For example, if asked to design a framing for a Nike basketball shoe, a blurb highlighting their collaboration with a non-existent NBA team or player would be incorrect. In general, the \emph{soundness} may be more nuanced than a binary outcome; the framing could take certain liberties that, while not blatantly incorrect, may be undesired. As such, we propose using an LLM to score soundness with a set of values between 0 and 1. This module is given in-context information about the instance along with the generated framing. This score is part of the feedback to the framing-generating LLM.

\textbf{Computing sender utility:} For any given instance with a receiver prior belief $\mu_c$, the corresponding optimal signaling scheme $\pi^*(\mu_c)$ can be analytically computed by solving the linear program specified in Equations~\eqref{equation:joint_lp} in poly-time.

\textbf{Generating framing:} Building on the success of in-context learning via \emph{textual gradients} \citep{pryzant2023automatic}, we propose that an LLM generate a framing based on instance-relevant information and a task description, and iteratively refine it through feedback. The task description defines key parameters for the framing, such as word count and style, while also outlining the feedback to be expected. For each generated framing, the induced prior belief is estimated and then used to compute the corresponding sender utility; this is scaled by the soundness score. This final quantitative score is supplemented with the reasoning behind the generated belief and soundness score to construct the feedback string. Alongside this, the feedback also contains a \emph{textual gradient} where we outline how the induced belief and utility changed between the last generated framing and the current one. The LLM context maintains a window of the last $k$ generated framings, their feedback, and the best candidate framing generated so far.

\subsection{An Advertising Case Study}\label{subsec:case_study}
To evaluate our proposed empirical framework, we conduct a case study focused on advertising. We consider a fictional clothing brand, \emph{Himalaya}, known for durable, high-quality gear favored by outdoor enthusiasts. The brand is launching a new outerwear line targeting the style-conscious, casually active athleisure market. Such consumers primarily care about looking good but like the idea of performance wear. The brand's advertising campaign aims to attract this new demographic without compromising brand identity and includes: (1) a new slogan, (2) a description of the new product line, and (3) discount offers. The first two constitute the framing since they are chosen \emph{a priori} and apply to all products; for the third, the sender chooses how the discounts are correlated with product features, and thus represents the signaling scheme.

There are four product categories based on two mutually exclusive attributes: \emph{(Chic vs Functional)} and \emph{(Durable vs Not Durable)}. Buyers can choose from three actions: ``buy on sale'', ``buy at regular price'', or ``not buy'', and there are three corresponding signals: ``Ad with Discount'', ``Ad without Discount'', and ``No Ad''. Both the user and brand receive zero utility from not buying. As a realistic constraint, the very desirable (Chic, Durable) set of products are never on sale\footnote{This also ensures there is no dominant action for either party, making the instance non-degenerate.}, but still provide positive utility to both parties if purchased at regular prices. The consumer would prefer not buying clothes that are not 
chic, followed by buying them on sale. The brand, on the other hand, would always prefer to sell at regular prices over discounts, over not selling at all. The utilities matrices are 
specified in Appendix~\ref{appendix:advertising_exp}.
The appendix also contains the exact prompts used to describe the target demographic, and the description of the brand and its product features, which is taken directly from 
a real-world product line. For all LLM experiments, we used the default temperature setting.


\begin{figure}[t]
    \centering
    \begin{minipage}{0.45\textwidth}
        \centering
        \includegraphics[width=\textwidth]{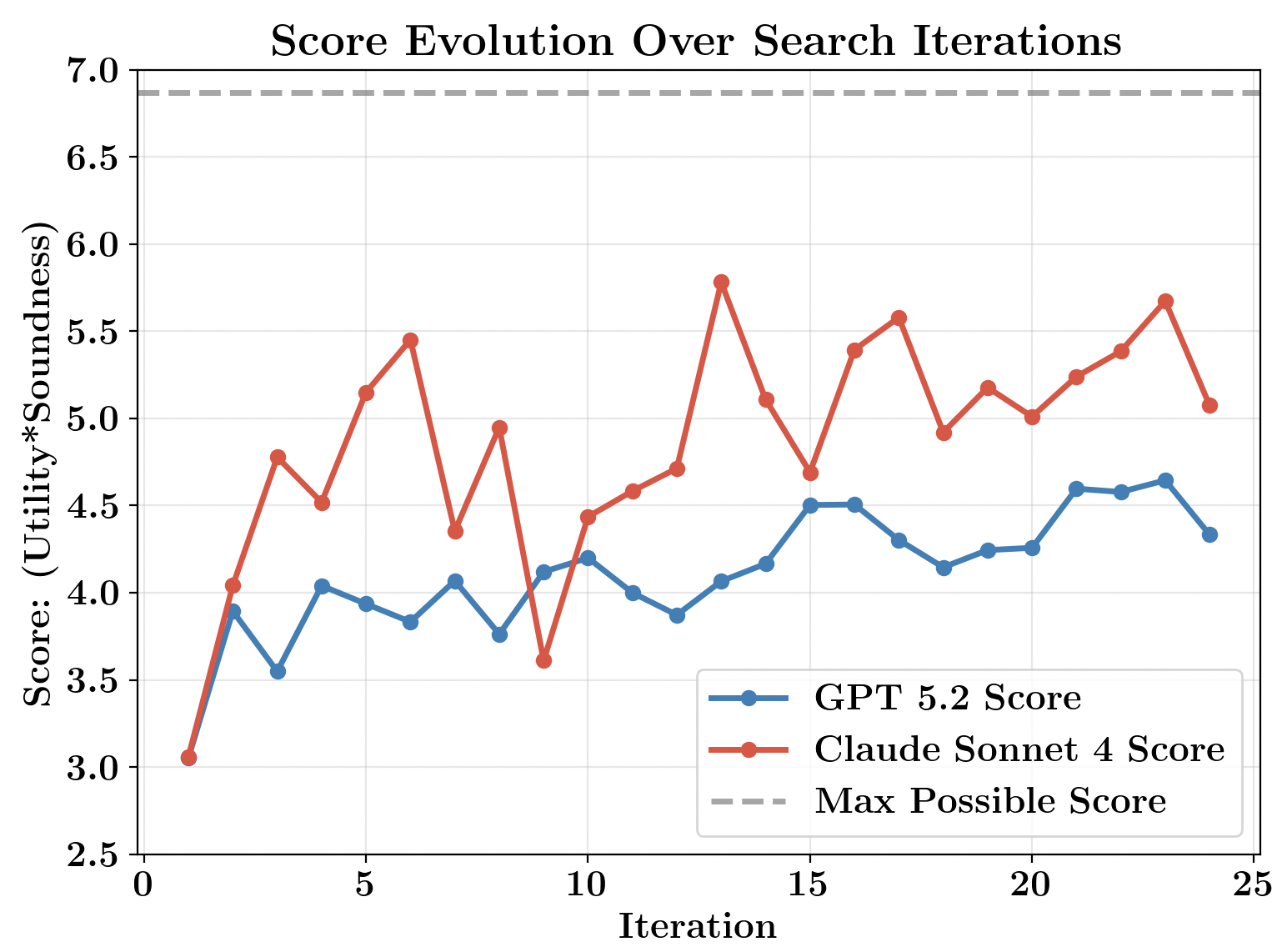}
        \vspace{-2.0em}
        \caption{Evolution of the scores (product of sender utility and soundness) from LLM generated framing. For a framing $c$ we determine the induced belief $\mu_c$ by an LLM, and then compute the sender utility for it under optimal signaling.}
        \label{fig:score_evolution}
    \end{minipage}
    \hfill
    \begin{minipage}{0.53\textwidth}
        \centering
        \includegraphics[width=\textwidth]{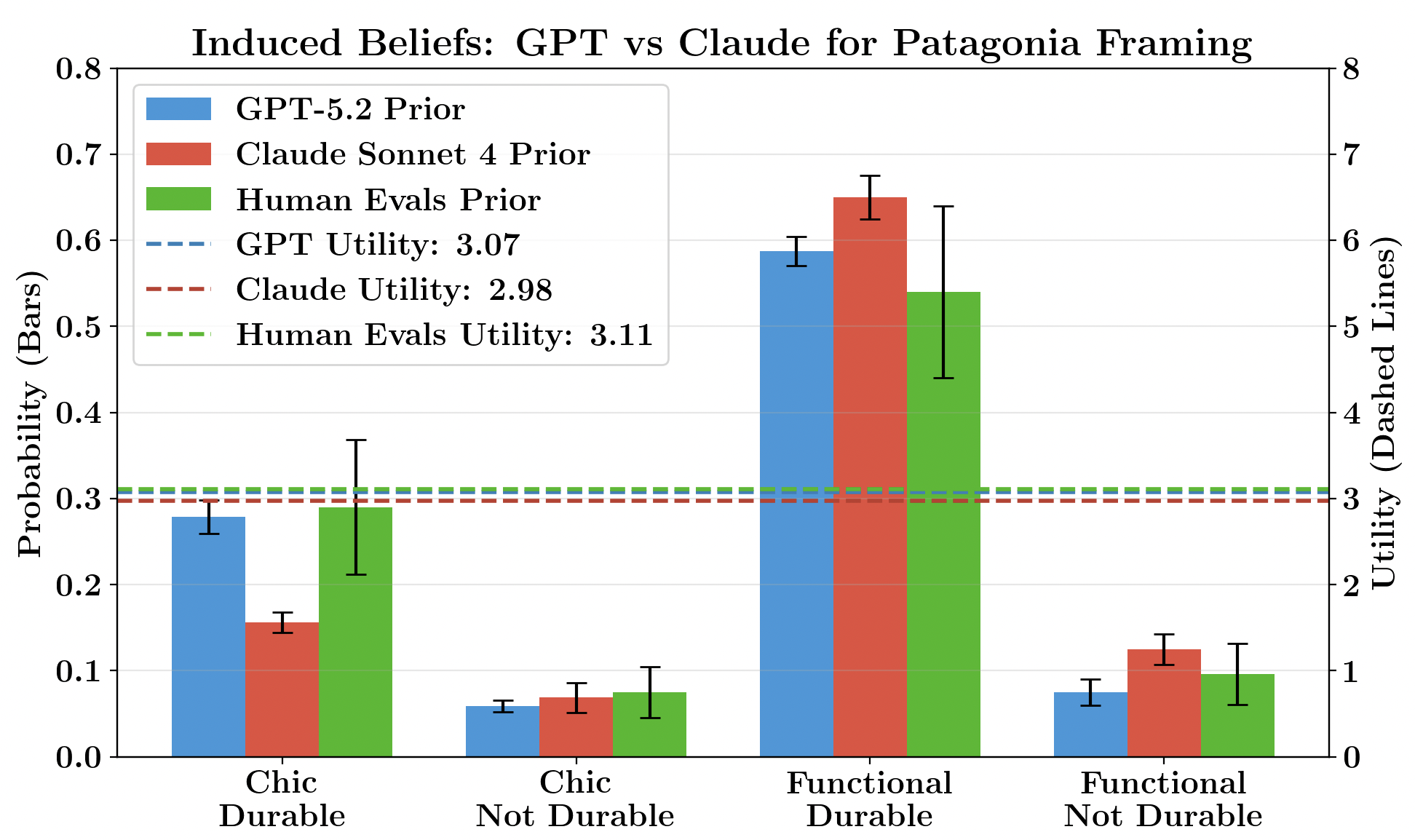}
        \caption{The induced beliefs generated by Claude and GPT models for the framing (slogan and brand description) of a real-world brand (Patagonia).
        They are compared against how humans (through Prolific) evaluated the same framing. Dashed lines are the sender utility under those beliefs. All results are presented with 95\% confidence intervals.}
        \label{fig:induced_priors}
    \end{minipage}
\end{figure}

\paragraph{\bf Estimating the Framing Induced Belief: }\label{subsubsec:belief_estimation}
We conduct a preliminary investigation to evaluate how well LLMs can estimate the framing-induced belief.
The framing for this exploration is the brand slogan and product description that a real-world brand (Patagonia) uses to describe its outerwear lineup, which happens to not include the brand name.\footnote{See Appendix~\ref{app_sec:llm_prompts} for the exact slogan and description as well as full prompts.}
GPT-5.2 and Claude Sonnet 4 are then used to estimate the induced belief 
of our target demographic, following the procedure outlined in Section~\ref{sec:methodology}. We elicit the belief 10 times per model and plot the average beliefs and 95\% confidence intervals over the four states in Figure~\ref{fig:induced_priors}.

To validate these LLM-generated beliefs against human judgments, we recruit 55 participants from Prolific and provide them with the same context given to the LLMs: a description of the target demographic, the four possible product states, and the framing. While one approach would be to directly elicit numerical beliefs from humans (same way as done for LLMs), a body of work in survey methodology and judgment under uncertainty finds that untrained respondents struggle to provide well-calibrated probability distributions~\citep{lichtenstein1981calibration, tversky1982judgment}. This is particularly the case when there are several outcomes due to the cognitive difficulty of ensuring consistency and proper normalization. To sidestep this issue, we ask participants to make pairwise comparisons. Specifically, given two possible states (e.g., ``Chic \& Durable'' vs ``Chic \& Not Durable''), we ask the participant which of the two is a more likely descriptor of the products from this brand from the perspective of the target consumer. Eliciting all possible pairwise comparisons then allows us to fit the classical Bradley-Terry model~\citep{bradley1952rank}, 
which posits that for two options $i,j$, $\Pr(i > j) = \frac{\exp(\lambda_i)}{\exp(\lambda_i) + \exp(\lambda_j)}$, where the $\lambda$ parameters represent numerical preferences.
We determine $\lambda$ by maximizing the log-likelihood given the participants' pairwise comparison responses, 
which after normalization yields a valid probability distribution (belief over the 4 states).
We plot the human belief 
with green bars and 95\% confidence intervals in Figure~\ref{fig:induced_priors}.

We observe that the LLM-generated beliefs, especially that of GPT-5.2, are very close to the aggregate human belief.
The Claude belief also induces the same rankings over the states as the human belief. While the difference between Claude's and human's beliefs may seem somewhat significant in the belief space, it is much less so in the sender utility space. That is, if we compute the sender utility under optimal signaling scheme for each belief (plotted in the dashed lines), they are very similar. Beyond the similarity between beliefs, we observe that the LLM generated beliefs tend to have less variance 
than the human evaluated beliefs. Given the excellent performance of GPT-5.2, we end up using this to estimate induced beliefs for the end-to-end experiment. 

\paragraph{\bf Optimizing in the Framing Space: } With the framing-to-belief mapping approximated by GPT-5.2, we run our end-to-end framework to jointly optimize framing and signaling for this case study. For each generated framing, we compute both the soundness score and the sender utility under optimal signaling for that induced belief. The product of soundness and utility (denoted as the score), along with a \emph{textual gradient} of how the current framing influences beliefs and scores compared to the last one, is provided as in-context feedback. We experiment with both GPT-5.2 and Claude Sonnet 4, 
with their respective scores plotted as a function of the iteration in Figure~\ref{fig:score_evolution}. Both models start with the framing that is currently used by Patagonia (and used in the induced belief estimation experiments above).\footnote{The Patagonia slogan and product description do not reference the brand name. 
And we do not explicitly mention ``Patagonia'' in any of the LLM experiments. 
}
We note that Claude tends to have aggressive changes in score. This is consistent with its generation traces, which show substantial exploration and modification between rounds. This leads to Claude achieving a better score than GPT-5.2, which tends to take a more conservative approach. Indeed, the GPT-5.2 traces show light framing modification between rounds. The best generated framings are given in Appendix~\ref{app_sec:optimal_framings}.

We conclude with two observations. First, the gray dashed line in Figure~\ref{fig:score_evolution} refers to the maximum achievable score assuming that the entire simplex is inducible, i.e., $B = \Delta(\Omega)$.  This requires a belief that is highly concentrated in the ``functional and durable'' class. While the LLM framings also optimize toward that belief, they are limited by what is actually inducible in the receiver through framing. Second, the fact that the LLM framings outperform the starting framing used by Patagonia should not be interpreted as evidence that Patagonia's 
current advertising strategy
is ineffective.
Patagonia's framing is likely well-suited to its core customer base, but not for the fashion-conscious athleisure demographic we target in this case study. 
This underscores the nuanced and context-dependent nature of framing.

\section{Discussion}\label{section:discussion}

This paper bridges classical Bayesian signaling with insights from behavioral economics and psychology, which emphasize that linguistic and contextual framing plays an important role in shaping the beliefs and perceptions of decision-makers.
Historically, the study and optimization of framing 
required costly and time-intensive methods such as focus groups, limiting systematic exploration.
The emergence of large language models (LLMs) offers a practical alternative: they enable cost-effective generation, evaluation, and iterative refinement of framing. 
Our empirical results demonstrate the viability of this approach, while our theoretical analysis characterizes the resulting optimization landscape.

From a theoretical perspective, we show that framing can be extremely powerful -- small belief shifts induced by framing may lead to large gains in sender utility -- but also fundamentally challenging to optimize. In particular, framing-only optimization under a fixed signaling scheme is computationally intractable and sensitive to approximation error.
In contrast, the joint optimization of framing and signaling is substantially more benign: the sender's optimal utility becomes continuous in the induced belief, admits efficient approximation algorithms, and is robust to errors in belief estimation. This distinction precisely identifies when algorithmic, LLM-assisted persuasion is promising and when it faces intrinsic barriers.

Empirically, we take first steps toward a systematic, algorithmically driven alternative to human-crafted persuasive communication. Our framework combines LLM-based search over natural-language framing with analytical optimization of signaling schemes, resulting in a simple and practical procedure for co-designing textual framings and quantitative signals.
Nevertheless, our results should not be interpreted as suggesting that existing real-world advertising strategies are ineffective; rather, they highlight the context-dependent nature of framing and its interaction with target demographics.

This work opens several avenues for future research. On the theoretical side, the exact computational complexity of framing–signaling joint optimization remains open.
Empirically, further study is needed to better understand how closely LLM-generated beliefs align with human beliefs, or how satisfied humans are with LLM-delegated decision-making in signaling settings.
Finally, beyond Bayesian persuasion, it would be fruitful to investigate other influential communication models, such as cheap talk \citep{crawford_strategic_1982, farrell1996cheap} and mediation \citep{roger1991game}, to further broaden the scope of information design with behavioral framing effects.




\bibliographystyle{plainnat}
\bibliography{bibliography}

\newpage
\appendix 

\newpage
\appendix
\section{Omitted Proofs in Section \ref{sec:context_only_optimization}}
\label{appendix:context_only}
\subsection{Proof of Proposition \ref{prop:no_randomization}}
    Let $P_C \in \Delta(C)$ be a distribution over the framing space. Consider the sender's expected utility under this distribution (we consider the framing space to be discrete here, but the result immediately holds for the continuous setting by replacing $\sum_c$ with $\int_c$):
    \begin{align*}
        \E_{P_C}\big[u(a^*_{c, \pi, s}, \omega)\big] 
        &= \sum_{c \in C}{P_C(c)}\sum_{\omega \in \Omega} \sum_{s \in \Ss} \mu_0(\omega){\pi(s|\omega)u(a^*_{c,\pi, s}, \omega)}. 
    \end{align*}
    Let $c_{\max} = \argmax_{c \in C}{\sum_{\omega}\sum_{s}{\mu_0(\omega)\pi(s|\omega)u(a^*_{c, \pi, s}, \omega)}}$. Then it is clear that using framing $c_{\max}$ upper-bounds the expected utility achieved from the randomized strategy. Formally:
    \begin{equation*}
        \E_{P_C}\big[u(a^*_{c,\pi, s}, \omega)\big] \leq\sum_{\omega\in\Omega}\sum_{s\in \Ss}{\mu_0(\omega)\pi(s|\omega)u(a^*_{c_{\max}, \pi, s}, \omega)} = \E\big[u(a^*_{c_{\max}, \pi, s}, \omega)\big].  
    \end{equation*}
    So, there is no need to use randomized framing strategies.

\subsection{Proof of Theorem \ref{thrm:np_hardness}}
\label{proof:np-hardness}
    We will show that finding the optimal utility in a specific class of Bayesian Stackelberg games (BSG) can be reduced to our problem of computing the optimal sender utility by optimizing only framing/receiver prior (OF). \citet{conitzer2006computing} prove that the former problem is NP-Hard. Specifically, it is hard to compute the optimal utility for the following class of BSG problems, which they show is sufficient to ensure that the Independent Set problem can be reduced to it. 
    \squishlist
        \item Follower has binary actions ($a_0, a_1$) with positive bounded utility: $u^{\theta}_f(a_{\ell}, a_f) \in [0, v^{max}]$, where $\vmax \leq |\A_{\ell}|$.
        \item Follower always has utility $1$ for $a_0$. That is, $\forall \theta, a_{\ell}$, $u^{\theta}_f(a_{\ell}, a_0) = 1$.
        \item Leader utility is binary and does not depend on the leader's action (only the follower): $u_{\ell}(a_f) \in \{0, 1\}$.
        \item The least probable type occurs with non-zero probability: $\min_{\theta}P(\theta) \triangleq P_{min} \geq 0$.
    \squishend

    For any given instance of the BSG   with the above characteristics, denoted by $\IBS$, with optimal solution $x^*$ achieving optimal leader utility $\BS(\IBS, x^*)$, we will give a poly-time construction of an OF problem instance $\IOC'$ whose optimal solution $\mu_c^*$ is such that $\OC(\IOC', \mu^*) = \BS(\IBS, x^*)$. Hence if the OF problem can be efficiently solved, it would imply efficient solving of the class of BSG problem described above, which is known to be NP-Hard. For a given instance $\IBS = (\Theta, \A_{\ell}, \A_{f}, P(\theta), u_{\ell}, u_f)$,  we first construct an intermediate instance $\IOC = (\Omega, C, \A, S, \mu_0, u, v, \pi)$ as follows:
    \squishlist
        \item The state space for $\IOC$ is: $\Omega = \{\omega_{\al}\}_{\al \in \A_{\ell}} \cup \{\omega_{\theta}\}_{\forall \theta \in \Theta} \cup \wtilde$
        \item The receiver's action space is: $\A = \{\atheta_0\}_{\forall \theta \in \Theta} \cup \{\atheta_1\}_{\forall \theta \in \Theta} \cup \atilde_1 \cup \atilde_2$
        \item The signal space is: $S = \{s_{\theta}\}_{\theta \in \Theta}$.
    \squishend
    
    For this instance and $\varepsilon > 0$, we specify the sender's prior $\mu_0$, the fixed signaling scheme $\pi$ and the sender and receiver utilities $u, v$ as follows:
    \squishlist
        \item Prior $\mu_0$: $\mu_0(\wtilde) = 1 - \varepsilon$ ; $\forall \theta, \, \mu_0(\omega_{\theta}) = \frac{\varepsilon}{|\Theta|}$; $\forall \al, \, \mu_0(\omega_{\al}) = 0$.\footnote{We let $\mu_0(\omega_{\al}) = 0$ for convenience. One can also construct an instance with a small positive $\mu_0(\omega_{\al}) > 0$. }
        \item Signaling $\pi$: $\forall \theta, \pi(s_{\theta} | \wtilde) = P(\theta)$, $\pi(s_{\theta} | \omega_{\theta}) = 1$; $\forall \theta \ne \theta' ,\, \pi(s_{\theta'} | \omega_{\theta}) = 0$ ; $\forall \theta, \al, \, \pi(s_{\theta} | \omega_{\al}) = \frac{1}{|\Theta|}$. 
        \item Sender Utility $u(a, \omega)$:
        \squishlist
            \item $\forall \al, u(\atheta_*, \omega_{\al}) = u_{\ell}(a_*)$, where $* \in \{0, 1\}$. 
            \item $\forall \theta, u(a^{\theta'}_{*}, \omega_{\theta}) = -L$ if $\theta' \ne \theta$ ; otherwise $u(a^{\theta}_{*}, \omega_{\theta}) = 0$. 
            \item $\forall \omega, u(\atilde_1, \omega) = -N$ ; $u(\atilde_2, \omega) = -K$
            \item $u(\atheta_*, \wtilde) = u_{\ell}(a_*)$, for all $\theta$. 
        \squishend
        \item Receiver Utility $v(a, \omega)$:
        \squishlist
            \item $\forall \al, v(\atheta_*, \omega_{\al}) = u^{\theta}_f(a_{\ell}, a_*)$ ; $v(\atilde_1, \omega_{\al}) = -M - 1$ ; $v(\atilde_2, \omega_{\al}) = 0$
            \item $\forall \theta, v(a^{\theta'}_{*}, \omega_{\theta}) = -M$ for $\theta \ne \theta'$ ; $v(a^{\theta}_{*}, \omega_{\theta}) = 0$ ; $v(\atilde_1, \omega_{\theta}) = - M - 1$ ; $v(\atilde_{2}, \omega_{\theta}) = +K$
            \item $v(\atilde_1, \wtilde) = +N$ ; $v(a \ne \atilde_1, \wtilde) = 0$. 
        \squishend
    \squishend

    \noindent The high-level intuition for this instance is as follows. When the receiver sees a signal $s_{\theta}$ (which is proxying type $\theta$ in BS), we want them to only consider actions $\atheta_0, \atheta_1$, which directly corresponds to follower utility of type $\theta$ in BS. Receiver utility in O, however, does not explicitly depend on $\theta$, but rather on the state $\omega$. Hence we expand the state space to include $\omega_{\theta}$ states. Using a fixed signaling scheme, we want to ensure that $\mu_c$ always induces a slight belief in the receiver that states $\omega_{\theta}$ occurred. The sender is incentivized to do this since otherwise, the receiver could take $a^{\theta'}_*$ actions at state $\omega_{\theta}$ (which occurs with non-zero probability), which is very bad for the sender. They also don't want to put too much weight on $\omega_{\theta}$ states, lest the receiver take the bad (for sender) $\atilde_2$ action. Lastly, we add an additional state $\wtilde$ to ensure the OF objective captures the BSG objective, which depends on the type. Formally, the OF optimization problem under this instance construction above can be written as:
    \begin{align}
        & \maximize_{\mu_c} \,\, \underbrace{(1-\varepsilon)\sum_{\theta}{P(\theta)u(a^*(\mu_c, s_{\theta}))}}_{\text{state $\wtilde$}} + \underbrace{\frac{\varepsilon}{|\Theta|}\sum_{\theta}{u(a^*(\mu_c, s_{\theta}), \omega_{\theta})}}_{\text{for states $\omega_{\theta}$ where $\pi(s_{\theta}|\omega_{\theta}) = 1$}}\\
        & \text{s.t} \quad a^*(\mu_c, s_{\theta}) = \argmax_{a \in \A}{\left[\mu_c(\omega_{\theta})v(a, \omega_{\theta}) + P(\theta)\mu_c(\wtilde)v(a, \wtilde) + \frac{1}{|\Theta|}\sum_{\omega_{\al}}{\mu_c(\omega_{\al})v(a, \omega_{\al})}\right]}
    \end{align}
    We now prove three intermediate results that will specify the necessary relations between the constants used in our $\IOC$ instance and disentangle the key arguments needed for the reduction. 

    \begin{lemma}
        If $\mu_c(\omega_{\theta}) = \frac{v^{max}}{|\Theta|M}, \, \forall \theta$, $\mu_c(\wtilde) = 0$, with $\vmax \leq \frac{M}{1+K}$, then (1) the receiver always chooses between the two action $\{\atheta_0, \atheta_1\}$ on receiving signal $s_{\theta}$ and (2) the sender utility is at least 0.
    \end{lemma}
    \begin{proof}
         Since $\mu_c(\wtilde) = 0$, we need not consider the receiver taking action $\atilde_1$, since it is dominated by some other action at all remaining states. We first show that on some signal $s_{\theta}$,  they will never take action $\atilde_2$. Indeed, it is obedient for the receiver to take action $\atheta_0$ as opposed to $\atilde_2$ on receiving a signal $s_{\theta}$:
        \begin{align}
            & \mu_c(\omega_{\theta})[v(\atheta_0, \omega_{\theta}) - v(\atilde_2, \omega_\theta)] + \frac{1}{|\Theta|}\sum_{\al}{\mu_c(\omega_{\al})[v(\atheta_0, \omega_{\al}) - v(\atilde_2, \omega_{\al})]} \\
            &= - \mu_c(\omega_{\theta})K + \frac{1}{|\Theta|}\sum_{\omega_{\al}}{\mu_c(\omega_{\al})v(\atheta_0, \omega_{\al})} = \frac{1}{|\Theta|}\left(1 - \frac{\vmax}{M}\right) - \frac{K\vmax}{|\Theta|M} \geq 0
        \end{align}
        where the second equality in the second line follows since at state $\omega_{\al}$, the receiver utility matches that of the BSG setting - i.e $v(\atheta_0, \omega_{\al}) = u^{\theta}_f(a_0, \al)$ - and in the BSG instances we care about, the receiver always gets utility 1 by taking action $a_0$, $v(\atheta_0, \omega_{\al}) = 1$ for all $\omega_{\al}$. This is greater than or equal to 0 due to our choice of constants satisfying $\vmax \leq \tfrac{M}{1+K}$ (we assume ties break in favour of $\atheta$ actions). Thus the receiver will not take action $\atilde_2$ on any signal $s_{\theta}$. 
    
        \noindent Next, we show that the receiver will not take any ``incorrect type'' actions $a^{\theta'}_*$ on receiving signal $s_{\theta}$. Suppose by contradiction they take a deviating action $a^{\theta'}_*$. Then they can expect a utility of at most $\frac{\vmax}{|\Theta|} - \frac{\vmax}{|\Theta|} = 0$. But we know they can achieve a utility of at least $1$ by playing $\atheta_0$ on each signal $s_{\theta}$. Thus, under the given specifications of $\mu_c$, the receiver will always play action $\atheta_*$ on signal $s_{\theta}$. Since the sender's utility on such actions is always at least 0 (mainly due to BSG instance having binary leader utility), the sender achieves at least 0 expected utility under this $\mu_c$.
    \end{proof}

    \begin{lemma}
        Let $\varepsilon \in (0,1)$, $L > \frac{|\Theta|}{\varepsilon}$, $\vmax \leq \frac{M}{1+K}$, and $N, K > \frac{1}{(1-\varepsilon)P_{min}}$. Then for an optimal solution $\mu_c^*$, the receiver only takes actions from $\{\atheta_0, \atheta_1\}$ when receiving signal $s_{\theta}$. This holds even if sender utilities are scaled by a positive constant.
    \end{lemma}
    \begin{proof}
        We partition the cases where this does not hold into three cases and for each, we indicate the suboptimality of $\mu^*_c$ with respect to a feasible solution that does conform to the above.\\

        (1) \emph{$\exists$ a signal $s_{\theta}$ where the receiver takes action $\atilde_1$}. If this were to occur, the sender utility is at most (note that the max sender utility in our $\IOC$ instance is 1 and in states $\omega_\theta$, the maximum utility is 0):
        \begin{align}\label{eq:np_hard_claim_2_n}
            \underbrace{-N(1-\varepsilon)P(\theta)}_{\text{on $\wtilde$ and signal $\theta$}} + \underbrace{(1-\varepsilon)}_{\text{on $\wtilde$ and other signals}} - \underbrace{\frac{N\varepsilon}{|\Theta|}}_{\text{on $\omega_{\theta}$ and $s_{\theta}$}} \leq -N(1-\varepsilon)P(\theta) + 1 < \frac{-P(\theta)}{P_{min}} + 1 \leq 0
        \end{align}
        where the last inequality arises from substituting the lower bound of $N$ specified. The sender thus achives negative utility. However, using claim 1, we know of a feasible specification of $\mu_c$ under these parameters where the sender can achieve at least 0 utility. Thus the $\mu_c^*$ here cannot be optimal. Note that when we scale by a positive constant, the last part of Eq. \eqref{eq:np_hard_claim_2_n} simply becomes $c\left[\frac{P(\theta)}{P_{min}} + 1\right] \leq 0$ for the same reason as above.

        (2) \emph{$\exists$ a signal $s_{\theta}$ where the receiver takes action $\atilde_2$.} As before, if this were to occur, the sender utility for this $\mu^*_c$ is at most:
        \begin{equation}
            -K(1-\varepsilon)P(\theta) + (1-\varepsilon) - \frac{K\varepsilon}{|\Theta|} \leq -K (1-\varepsilon) P(\theta) + 1 \leq \frac{-P(\theta)}{P_{min}} + 1 \leq 0
        \end{equation}
        where in the last inequality, we substitute the lower bound of $K$ specified. As before, the sender achieves negative utility, even through claim 1 shows it is possible to achieve a utility of 0, indicating suboptimality. Further, it is impervious to positive scaling of sender utilities.

        (3) \emph{$\exists$ a signal $s_{\theta}$ where the receiver takes an action $a^{\theta'}_*$.} If this were to occur, consider the sender utility:
        \begin{align}
            & (1-\varepsilon)\underbrace{\sum_{s_{\theta}}{P(\theta)u(a^*(\mu_c, s_{\theta}))}}_{\text{at most 1}} + \frac{\varepsilon}{|\Theta|}\underbrace{u(a^{\theta'}_*, \omega_{\theta})}_{-L} + \frac{\varepsilon}{|\Theta|}\sum_{s_{\hat{\theta}}}{\underbrace{u(a^*(\mu_c, s_{\hat{\theta}}), \omega_{\hat{\theta}})}_{\text{at most } 0}} \\
            & \leq (1 - \varepsilon) - \frac{L \varepsilon}{|\Theta|} \leq 1 -  \frac{L \varepsilon}{|\Theta|} < 0 \label{eq:np_hard_claim_2_l}
        \end{align}
        where the last inequality follows since $\frac{|\Theta|}{\varepsilon} < L$. Again the sender receives negative utility when it is possible to achieve at least 0 utility due to claim 1. As before, if we were to scale by a positive constant $c$, inequality \eqref{eq:np_hard_claim_2_l} simply becomes $c\left[(1-\varepsilon) - \frac{L\varepsilon}{|\Theta|}\right] < 0$ which still becomes negative due to the choice of $L$.  
    \end{proof}
        
    \begin{lemma}
        Let $\varepsilon \in (0,1)$, $L > \frac{|\Theta|}{\varepsilon}$, $\vmax \leq \frac{M}{1+K}$, and $N, K > \frac{1}{(1-\varepsilon)P_{min}}$. Then for an optimal solution $\mu_c^*$, we can construct a solution $\mu'$ in poly-time such that $\OC(\IOC, \mu^*) = \OC(\IOC, \mu')$, $\mu'(\wtilde) = 0$, $a^*(\mu'_c, s_{\theta}) \in \{\atheta_1, \atheta_0\}$. This holds even when all sender utilities are scaled by a positive constant.
    \end{lemma} 
    \begin{proof}
        From claim 2, we already know that $\mu^*_c$ satisfies $a^*(\mu^*_c, s_{\theta}) \in \{\atheta_1, \atheta_0\}$. We now show that any weight $\mu^*_c$ places on $\mu(\wtilde)$ can be shifted without changing this invariant. For each signal $s_{\theta}$, let $a_{\theta}$ denote the receiver's optimal action for this signal. Then the receiver's obedience for $a_{\theta}$ implies:
        \begin{equation}\label{eq:claim_3_eq}
            -\mu_c(\wtilde)P(\theta)v(a', \wtilde) - \mu_c(\omega_{\theta})v(a', \omega_{\theta}) + \frac{1}{|\Theta|}\sum_{\al}{\mu_c(\omega)}[v(\atheta, \omega_{\al}) - v(a', \omega_{\al})] \geq 0 \,\, \forall a'
        \end{equation}
        Now consider a $\mu'_c$ where $\mu'_c(\wtilde) = 0$ and $\mu'_c(\omega \ne \wtilde) = \frac{1}{1 - \mu_c^*(\wtilde)}\mu_c^*(\omega)$. This is clearly a valid distribution since $\sum{\mu_c(\omega)} = \frac{1}{1 - \mu_c^*(\wtilde)}\sum{\mu_c^*(\omega)} = 1$. When $a' = \atilde_1$, since the invariant is originally maintained and $v(\atilde_1, \wtilde) = +N$, the negative first term in Eq.~\eqref{eq:claim_3_eq} becomes 0 and the last two terms (which together must have been positive) are just increased in scale. Hence the invariant is maintained. For any $a' \ne \atilde_1$, the first term is 0 in Eq.~\eqref{eq:claim_3_eq}, and the adjusted $\mu'_c$ simply scales the remaining two terms which must be non-negative. Hence the invariant is always maintained. In other words, $a^*(\mu'_c, s_{\theta}) = a^*(\mu^*_c, s_{\theta}) \in \{\atheta_0, \atheta_1\}$. Lastly, since the choice of $\mu_c$ only affects the sender through the decision taken by the receiver, and both $\mu^*_c$ and $\mu'_c$ lead the receiver to always behave in the same way, the sender utility is unchanged and the claim holds.
    \end{proof}

    We now prove BSG can be reduced to OF. For a BSG instance $\IBS = (\Theta, \A_{\ell}, \A_{f}, P(\theta), u_{\ell}, u_f)$, we construct an instance $\IOC = (\Omega, \A, S, \mu_0, \pi, u, v)$ as described earlier, in poly-time. Next, consider an instance $\IOC' = (\Omega, \A, S, \mu_0, \pi, \tfrac{1}{1-\varepsilon}u, v)$, which is identical to $\IOC$, except all sender utilities are now scaled by $\frac{1}{1-\varepsilon}$. Note that claims (1) and (3) depend purely on the receiver utility and sender utilities for $\atheta$ actions at $\omega_{\al}$ states being non-negative and the statement of (2) highlights that it holds when the sender utilities are scaled by a positive constant. In other words, all three claims hold on instance $\IOC'$. We now show that for any optimal $\mu^*_c$ to instance $\IOC'$, there exists a feasible $x'$ that achieves the same utility on the corresponding BSG instance. Similarly, for an optimal $x^*$ to $\IBS$, there exists a $\mu'_c$ that achieves the same utility on the corresponding OF instance. This naturally implies $BS(\IBS, x^*) = \OC(\IOC', \mu_c^*)$.

    \noindent \textbf{The ``$\Longrightarrow$'' direction:}
    Suppose we have an optimal $\mu^*_c$ for instance $\IOC'$; without loss of generality, we assume $\mu^*_c(\wtilde) = 0$ (if this is not the case, we can use Claim 3 to construct it to be so in poly-time). Since at each $s_{\theta}$, we are guaranteed that $a^*(\mu^*_c, s_{\theta}) \in \{\atheta_1, \atheta_0\}$, the sender utility is simply $(1 - \varepsilon)\sum_{\theta}{P(\theta)u(a^*(\mu^*_c, s_{\theta}))}$, which corresponds to the utility at state $\wtilde$ (note that $\mu_0(\wtilde)$ is not 0). For any $s_{\theta}$, without loss of generality, let $\atheta_1$ denote the optimal action. Then obedience with respect to $\atheta_0$ (the only other action possible since claim 3 disavows all others) implies:
    \begin{equation}
        \frac{1}{|\Theta|}\sum_{\omega_{\al}}{\mu^*_c(\omega_{\al})[v(\atheta_1, \omega_{\al}) - v(\atheta_0, \omega_{\al})]} + \mu^*_c(\omega_{\theta})\underbrace{[v(\atheta_1, \omega_{\theta}) - v(\atheta_0, \omega_{\theta})]}_{0} \geq 0
    \end{equation}
    Let $x' \in \Delta^{|\A_{\ell}|}$ be as follows: $x(\al) = \frac{1}{\sum_{\omega'_{\al}}{\mu^*_c(\omega'_{\al})}}\mu^*_c(\omega_{\al})$. Clearly this is a valid strategy since $\sum_{\al}{x(\al)} = 1$. Further, since this is just scaling of the $\mu_c^*(\omega_{\al})$ we have that:
    \begin{align}
        0 &\leq \sum_{\al}{x(\al)[v(\atheta_1, \omega_{\al}) - v(\atheta_0, \omega_{\al})]} = \sum_{\al}{x(\al)[u^{\theta}_f(a_1, \al) - u^{\theta}_f(a_0, \al)]}
    \end{align}
    This implies that the optimal action for a follower of type $\theta$ for strategy $x'$, $a^*_f(\theta, x) = a^*(\mu^*_c, s_{\theta})$, which is the optimal action for the OF receiver for the optimal \context $\mu^*_c$ and signal $s_{\theta}$. For $* \in \{0,1\}$, since the sender utility for $\atheta_*$ actions at the $\wtilde$ state in $\IOC$ is the same as the leader's utility for action $a^*$ in BS, and we are using $\IOC'$ where this sender utility is scaled by $\frac{1}{1-\varepsilon}$, we have that:
    \begin{equation}
        \OC(\mu^*_c) = (1-\varepsilon)\sum_{\theta}P(\theta)u(a^*(\mu_c^*, s_{\theta})) = \sum_{\theta}{P(\theta)u_{\ell}(a_f^*(x', \theta))} = BS(x')
    \end{equation}

    \noindent \textbf{The ``$\Longleftarrow$'' direction:}  Suppose we have an optimal solution to the $x^*$ to the BSG instance $\IBS$. Then by definition, the obedience condition holds for any type $\theta$ and the follower's optimal action. For an arbitrary type $\theta$, let the optimal receiver action be $a_1$ without loss of generality. Then:
    \begin{equation}\label{eq:ctx_only_IC_BS}
        \sum_{\al}{x^*(\al)[u^{\theta}_f(a_1, \al) - u^{\theta}_f(a_0, \al)]} \geq 0
    \end{equation}
    Now consider constructing $\mu'_c$ as follows. We first set $\mu'_c(\wtilde) = 0$ and $\mu'_c(\omega_{\theta}) = \frac{\vmax}{|\Theta|M}$ for all $\theta$. Due to claim 1, we already know that under this strategy, the receiver in the $\IOC$ instance will only choose between $\{\atheta_0, \atheta_1\}$ upon receiving a signal $s_{\theta}$ - in other words, we need not concern ourselves with actions $\atilde_1, \atilde_2$ or any $a_*^{\theta'}$, since these are dominated. Next, we set $\mu'_c(\omega_{\al}) = \left(1 - \frac{\vmax}{M}\right)x^*(\al)$. Observe that this is a valid distribution since $\sum_{\omega}\mu'_c(\omega) = \left(1 - \frac{\vmax}{M}\right) + \frac{\vmax}{M} = 1$. We then observe since Eq. \eqref{eq:ctx_only_IC_BS} holds for $x(\al)$, and $\mu'_c$ is simply a rescaling of $x(\al)$ on the $\omega_{\al}$ states, and $u_f^{\theta}(a_*, \al) = v(\atheta_*, \omega_{\al})$:
    \begin{equation}
         \sum_{\omega_{\al}}{\mu'_c(\omega_{\al})[v(\atheta_1, \omega_{\al}) - v(\atheta_0, \omega_{\al})]} \geq 0
    \end{equation}
    The expression is indeed sufficient to conclude that $\atheta_1$ is optimal for the OF instance receiver on getting signal $s_{\theta}$ since our construction of $\mu'_c$ ruled out all other actions except $a_*^{\theta}$. Since $u_{\ell}(\atheta_*, \al) = u_{\ell}(\atheta_*) = u(\atheta_*, \wtilde)$ in the $\IOC$ instance, and we are using $\IOC'$ where this sender utility is scaled by $\frac{1}{1-\varepsilon}$, we have that:
    \begin{align}
        \BS(x^*, \IBS) &= \sum_{\theta}{P(\theta)u_{\ell}(a^*_f(x, \theta))} = (1-\varepsilon)\sum_{\theta}{P(\theta)\frac{1}{(1-\varepsilon)}u_{\ell}(a^*_f(x, \theta))}\\
        &= (1-\varepsilon)\sum_{\theta}{P(\theta)u(a^*(\mu'_c, s_{\theta}))} = \OC(\mu', \IOC')
    \end{align}
    where the last equality follows from the fact that $\mu_0(\omega_{\al}) = 0$ and the receiver is always taking actions of type $\atheta_*$ on signal $s_{\theta}$, wherein we recall that $\pi(s_{\theta} | \omega_\theta) = 1$ sender utility $u(\atheta_*, \omega_{\theta}) = 0$.

    We have thus shown that the specific class of Bayesian Stackelberg games proven by \citet{conitzer2006computing} to be NP-Hard, can be expressed as an instance of the optimal \context problem, whose optimal solution exactly matches that of the BSG instance. The result of \cite{conitzer2006computing} in-fact, implies something stronger. They show that for a graph $G = (V, E)$, it is possible to construct a BSG instance of the type above such that the graph has an independent set of size $K$ if and only if the optimal leader utility in the BSG instance is at least $\frac{|E|}{|E| + 1} + \frac{K}{|V|(|E| + 1)}$.

    \noindent Their reduction uses $|E| + |V|$ types with the $P_{min} = \frac{1}{|V|(|E| + 1)}$. Since the sender utility is binary, there is no independent set of size $K$ if and only if the optimal leader utility $\leq \frac{|E|}{|E| + 1} + \frac{K-1}{|V|(|E| + 1)}$. This means that any $\frac{1}{2|V|(|E|+1)}$ additive approximation to the optimal leader utility would allow us to solve the $K$-Independent set problem, which is NP-Hard. Since they have $|E| + |V|$ and $|V|$ leader actions, we can formally state that it is NP-Hard to compute a $\frac{1}{2|\Theta||\A_{\ell}|}$ additive approximation to the BSG problem. 

    This additive approximation factor is predicated when the sender utility includes constant $L > \frac{|\Theta|}{\varepsilon}$ and $N,K \geq \frac{1}{(1-\varepsilon)P_{min}}$ for some $\varepsilon \in (0,1)$. To normalize this for utilities in the range $[0,1]$, we must divide by the range. If $N$ or $K$ dominates, then the range is $\frac{1}{(1-\varepsilon)P_{min}} + 1$ and any approximation constant must be greater than $\frac{1}{2|\Theta||\A_{\ell}|} \cdot \frac{(1-\varepsilon) P_{min}}{1 + (1-\varepsilon)P_{min}} \geq \frac{P_{min}(1-\varepsilon)}{4|\Theta||\A_{\ell}|}$. Now conversely, if $L$ dominates, then the range is $\frac{|\Theta|}{\varepsilon}  + 1$ and thus the approximation constant must be greater than $\frac{\varepsilon}{2|\Theta||\A_{\ell}|(\varepsilon + |\Theta|)} \geq \frac{\varepsilon}{4|\Theta|^2|\A_{\ell}|}$. In the optimal framing instance we construct for the reduction, $|\Theta| = |S|$ and $|\Omega| \geq |\A_{\ell}|$. Thus, it is NP-Hard to approximate the OF problem up to an additive $\min\left(\frac{P_{min}(1-\varepsilon)}{2|S||\Omega|},  \frac{\varepsilon}{4|\Theta|^2|\A_{\ell}|} \right)$ factor.

\subsection{Proof of Proposition \ref{prop:fixed-scheme-discontinuous}}
\label{proof:fixed-scheme-discontinuous}
Let $(u, v)$ be a pair of utility functions sampled from some continuous distribution. Recall that we consider receiver utility functions $v$ such that for every possible action, there is some belief in $\Delta(\Omega)$ such that this action is strictly optimal (inducible). Consider any pair of actions $a_1, a_2 \in A$.  Since $a_1$ is strictly inducible, there must be some state $\omega_1 \in \Omega$ under which $v(a_1, \omega_1) > v(a_2, \omega_1)$.  Since $a_2$ is strictly inducible, there must be some state $\omega_2 \in \Omega$ under which $v(a_1, \omega_2) < v(a_2, \omega_2)$. This means that, if the receiver's prior $\mu$ is deterministically on $\omega_1$, then we have 
\begin{align*}
    \sum_{\omega \in \Omega} \mu(\omega) \pi(s_0 | \omega) \Big( v(a_1, \omega) - v(a_2, \omega) \Big) ~ > ~ 0
\end{align*}
since $\pi(s_0|\omega_1) > 0$ by assumption. 
If the receiver's prior $\mu$ is deterministically on $\omega_2$, then we have 
\begin{align*}
    \sum_{\omega \in \Omega} \mu(\omega) \pi(s_0 | \omega) \Big( v(a_1, \omega) - v(a_2, \omega) \Big) ~ < ~ 0 
\end{align*}
since $\pi(s_0|\omega_2) > 0$ by assumption. 
Then, by the intermediate value theorem, there must exist a prior belief $\tilde \mu$ supported on $\{\omega_1, \omega_2\}$ only, namely, $\tilde \mu \in B_{\omega_1, \omega_2} = \{\mu \in \Delta(\Omega) \mid \mu(\omega_1) > 0, \mu(\omega_2)> 0,\forall \omega \notin\{\omega_1, \omega_2\}, \mu(\omega) = 0 \}$, and an action $a' \ne a_1$ such that the receiver is indifferent between $a_1$ and $a'$ upon receiving signal $s_0$: 
\begin{align*}
    0 & ~ = ~ \sum_{\omega \in \Omega} \tilde \mu(\omega) \pi(s_0 | \omega) \Big( v(a_1, \omega) - v(a', \omega) \Big) \\
    & ~ = ~ \tilde \mu(\omega_1) \pi(s_0|\omega_1) \Big( v(a_1, \omega_1) - v(a', \omega_1) \Big) + \tilde \mu(\omega_2) \pi(s_0|\omega_2) \Big( v(a_1, \omega_2) - v(a', \omega_2) \Big)
\end{align*}
and moreover $a'$ and $a_1$ are both weakly better than any other actions:
\begin{align*}
    a', a_1 \in \argmax_{a\in \A} \sum_{\omega \in \Omega} \tilde \mu(\omega) \pi(s_0|\omega) v(a, \omega). 
\end{align*}
Note that $a'$ may or may not be equal to $a_2$. Next, consider the receiver's best-response action $\tilde a^*_s$ upon receiving any signal $s \ne s_0$, under signaling scheme $\pi$ and prior $\tilde \mu$: 
\begin{align*}
    \tilde a^*_s \in \argmax_{a\in \A} \sum_{\omega \in \Omega} \tilde \mu(\omega) \pi(s|\omega) v(a, \omega). 
\end{align*}
Because $v$ is randomly sampled from a continuous distribution, and $\tilde \mu$ already made the receiver indifferent between $a'$ and $a_1$ at signal $s_0$, the probability that $\tilde \mu$ will make the receiver indifferent between any two actions under signal $s$ is $0$. So, $\tilde a_s^*$ must be unique for any $s \ne s_0$, with strict inequality
\begin{align*}
    \sum_{\omega \in \Omega} \tilde \mu(\omega) \pi(s|\omega) v(\tilde a^*, \omega) ~ > ~ \sum_{\omega \in \Omega} \tilde \mu(\omega) \pi(s|\omega) v(a, \omega), \quad \forall a \in \A\setminus\{\tilde a^*_s\}. 
\end{align*}
This means that, for sufficiently small $\eps > 0$, the receiver's best-response actions under the following two prior beliefs
\begin{align*}
    \tilde \mu^{+\eps} = (\tilde \mu(\omega_1)+\eps, \tilde \mu(\omega_2)-\eps, 0, \ldots, 0), \quad \quad \tilde \mu^{-\eps} = (\tilde \mu(\omega_1)-\eps, \tilde \mu(\omega_2)+\eps, 0, \ldots, 0)
\end{align*}
will still be $\tilde a^*_s$, given signal $s \ne s_0$. 

However, given signal $s_0$, because the receiver is indifferent between $a'$ and $a_1$ under prior $\tilde \mu$, the receiver will strictly prefer action $a_1$ under prior $\tilde \mu^{+\eps}$ and strictly prefer action $a'$ under prior $\tilde \mu^{-\eps}$, for sufficiently small $\eps > 0$. This means that the sender's utilities under priors $\tilde \mu^{+\eps}$ and $\tilde \mu^{-\eps}$ are 
\begin{align*}
    U_\pi(\tilde \mu^{+\eps}) & ~ = ~ \sum_{\omega \in \Omega} \mu_0(\omega) \Big( \sum_{s\in \Ss\setminus \{s_0\}} \pi(s|\omega) u(\tilde a^*_s, \omega) ~ + ~ \pi(s_0|\omega) u(a_1, \omega) \Big) \\
    U_\pi(\tilde \mu^{-\eps}) & ~ = ~ \sum_{\omega \in \Omega} \mu_0(\omega) \Big( \sum_{s\in \Ss\setminus \{s_0\}} \pi(s|\omega) u(\tilde a^*_s, \omega) ~ + ~ \pi(s_0|\omega) u(a', \omega) \Big). 
\end{align*}
We see that
\begin{align*}
    U_\pi(\tilde \mu^{+\eps}) - U_\pi(\tilde \mu^{-\eps}) & ~ = ~ \sum_{\omega \in \Omega} \mu_0(\omega) \pi(s_0|\omega) \Big( u(a_1, \omega) - u(a', \omega) \Big). 
\end{align*}
Because we assumed $\mu_0(\omega) > 0$, $\pi(s_0|\omega) > 0$, $\forall \omega \in \Omega$, and the randomly sampled utility function satisfies $u(a_1, \omega) \ne u(a', \omega)$ with probability $1$, we have 
\begin{align*}
    U_\pi(\tilde \mu^{+\eps}) - U_\pi(\tilde \mu^{-\eps}) & ~ = ~ C\ne 0 
\end{align*}
for some constant $C \ne 0$ independent of $\eps$.  This means that $U_\pi(\mu)$ is not continuous at $\tilde \mu$. 

\section{Omitted Proofs in Section \ref{sec:joint}}\label{appendix:joint}
\subsection{Proof of Observation \ref{ob:revelation_principle}}
\begin{proof}
    Consider an unrestricted signal space $S$, and for an instance $\I$, let $(c^*, \pi^*)$ denote the optimal strategy, with $\mu_c^*$ denoting the framing-induced belief. For this strategy, let $m: \A \rightarrow \Ss$ denote the correspondence between actions to signals under $(\mu_c^*, \pi^*)$. Then the sender utility is:
    \begin{equation}\label{eq:rev_principal_sender}
        \sum_{\omega \in \Omega}\sum_{a \in \A}u(a, \omega)\sum_{s \in m(a)}{\pi^*(s|\omega)}
    \end{equation}
    Consider a scheme $\pi'(a|\omega) = \sum_{s \in m(a)}{\pi(s|\omega)}$. We note that the receiver takes action $a$ when the receiver observes signal $a$ under this scheme since:
    \begin{align*}
        &  \sum_{\omega \in \Omega}{\mu^*_c(\omega)\pi^*(s|\omega)[v(a, \omega) - v(a', \omega)] \geq 0} \qquad \forall s \in m(a), ~ \forall a' \in \A \\
        & \implies \sum_{\omega \in \Omega}{\mu^*_c(\omega)[v(a, \omega) - v(a', \omega)] \sum_{s \in m(a)}{\pi^*(s|\omega)}} \geq 0 \qquad \forall a' \in \A
    \end{align*}
    It is thus clear that the sender utility from Eq. \eqref{eq:rev_principal_sender} in unchanged by using this direct scheme with signal space $\Ss$ equal $\A$ as action recommendations. 
\end{proof}

\subsection{Proof of Theorem \ref{thrm:joint_continuous}}
\label{proof:joint-continuous}
Without loss of generality, assume that the utility functions of the sender and the receiver are bounded: $\forall a \in \A, \forall \omega \in \Omega$, $u(a, \omega) \in [0, 1], v(a, \omega) \in [0, 1]$. Recall that $U^*(\mu)$ is the solution to the linear program outlined in \eqref{equation:joint_lp}. We aim to show that $U^*(\mu)$ is continuous at any $\mu \in \Delta(\Omega)$ satisfying $\mu(\omega) > 0, \forall \omega \in \Omega$. We break this result into a set of intermediate claims.

\begin{lemma}[Continuity of posterior]\label{lem:continuity-posterior}
Let $\pi : \Omega \to \Delta(\Ss)$ be any signaling scheme.  Let $\mu, \mu' \in \Delta(\Omega)$ be two receiver beliefs.  Let $\mu_s$, $\mu'_s$ be the posterior beliefs induced by signal $s$ under $\pi$ and priors $\mu$, $\mu'$ respectively.  Suppose $\min_{\omega \in \Omega} \mu(\omega) \ge p_0 > 0$.  Then, $\| \mu_s - \mu'_s \|_1 \le \frac{2}{p_0} \| \mu - \mu' \|_1$. 
\end{lemma}
\begin{proof}[Proof of Lemma \ref{lem:continuity-posterior}]
Let $\pi(s) = \sum_{\omega \in \Omega} \mu(\omega) \pi(s|\omega)$ and $\pi'(s) = \sum_{\omega \in \Omega} \mu'(\omega) \pi(s|\omega)$ be the probability of signal $s$ under prior $\mu$ and $\mu'$ respectively. By the definition of $\mu_s, \mu'_s$ and by triangle inequality, 
\begin{align*}
    \| \mu_s - \mu'_s \|_1 & = \sum_{\omega \in \Omega} \big| \tfrac{\mu(\omega) \pi(s|\omega)}{\pi(s)} - \tfrac{\mu'(\omega) \pi(s|\omega)}{\pi'(s)} \big| \\
    & \le \sum_{\omega \in \Omega} \big| \tfrac{\mu(\omega) \pi(s|\omega)}{\pi(s)} - \tfrac{\mu'(\omega) \pi(s|\omega)}{\pi(s)} \big| + \sum_{\omega \in \Omega} \big| \tfrac{\mu'(\omega) \pi(s|\omega)}{\pi(s)} - \tfrac{\mu'(\omega) \pi(s|\omega)}{\pi'(s)} \big|.
\end{align*}
For the first term above, 
\begin{align*}
    \sum_{\omega \in \Omega} \big| \tfrac{\mu(\omega) \pi(s|\omega)}{\pi(s)} - \tfrac{\mu'(\omega) \pi(s|\omega)}{\pi(s)} \big| = \sum_{\omega \in \Omega} \tfrac{\pi(s|\omega)}{\pi(s)} |\mu(\omega) - \mu'(\omega)|.
\end{align*}
We note that, $\forall \omega \in \Omega$, 
\begin{align}\label{eq:pi-ratio-le-p0}
    \tfrac{\pi(s|\omega)}{\pi(s)} = \tfrac{\pi(s|\omega)}{\sum_{\omega' \in \Omega} \mu(\omega') \pi(s|\omega')} &\le \tfrac{\pi(s|\omega)}{p_0 \sum_{\omega' \in \Omega}\pi(s|\omega')} \le \tfrac{1}{p_0}. \\
    \implies \sum_{\omega \in \Omega} \big| \tfrac{\mu(\omega) \pi(s|\omega)}{\pi(s)} - \tfrac{\mu'(\omega) \pi(s|\omega)}{\pi(s)} \big| &\le \sum_{\omega \in \Omega} \tfrac{1}{p_0} |\mu(\omega) - \mu'(\omega)| = \tfrac{1}{p_0} \| \mu - \mu' \|_1.  
\end{align}

For the second term, 
\begin{align*}
    \sum_{\omega \in \Omega} \big| \tfrac{\mu'(\omega) \pi(s|\omega)}{\pi(s)} - \tfrac{\mu'(\omega) \pi(s|\omega)}{\pi'(s)} \big| & = \sum_{\omega \in \Omega} \mu'(\omega) \pi(s|\omega) \big| \tfrac{\pi'(s) - \pi(s)}{\pi(s) \pi'(s)} \big| \\
    & = \sum_{\omega \in \Omega} \mu'(\omega) \pi(s|\omega) \big| \tfrac{\sum_{\omega'\in\Omega} (\mu'(\omega') - \mu(\omega')) \pi(s|\omega')}{\pi(s) \pi'(s)} \big| \\
    & \le \sum_{\omega \in \Omega} \mu'(\omega) \pi(s|\omega) \tfrac{\sum_{\omega'\in\Omega} |\mu'(\omega') - \mu(\omega')|  \cdot \max_{\omega'\in\Omega} \pi(s|\omega')}{\pi(s) \pi'(s)} \\
    & = \| \mu' - \mu \|_1 \sum_{\omega \in \Omega} \tfrac{\mu'(\omega) \pi(s|\omega)}{\pi'(s)} \tfrac{\max_{\omega'\in\Omega} \pi(s|\omega')}{\pi(s)} \\
    \text{by \eqref{eq:pi-ratio-le-p0}} ~ & \le \| \mu' - \mu \|_1 \sum_{\omega \in \Omega} \tfrac{\mu'(\omega) \pi(s|\omega)}{\pi'(s)} \tfrac{1}{p_0} = \tfrac{1}{p_0} \| \mu' - \mu \|_1.
\end{align*}

Therefore, we obtain $ \| \mu_s - \mu'_s \|_1 \le \tfrac{2}{p_0}\| \mu' - \mu \|_1$. 
\end{proof}

Recall that in the model (Section \ref{sec:model}) we assumed ``every action $a \in \A$ is strictly inducible'' in the receiver.  This means that there exists a constant $D>0$ such that, for every action $a \in \A$, there exists a belief $\eta_a \in \Delta(\Omega)$ for which $\E_{\omega \sim \eta_a}[v(a, \omega) - v(a', \omega)] \ge D > 0$ for every $a' \ne a$.  

We now want to show the following: \emph{Suppose the prior $\mu \in \Delta(\Omega)$ satisfies $\mu(\omega) \ge 2p_0 > 0, \forall \omega \in \Omega$.  Then, for any prior $\mu'$ satisfying $\| \mu' - \mu \|_1 \le \eps < \min\{p_0, \frac{p_0^2D}{2}\}$, we have:}
\[\big| U^*(\mu') - U^*(\mu) \big| \le \frac{4\eps}{p_0^2 D}.\]
This will directly prove the theorem. 

Let $\pi^*$ be the optimal signaling scheme for $\mu$, namely, a solution to the linear program in the definition of $U^*(\mu)$. 
Let $\pi^*(a)$ be the unconditional probability that $\pi^*$ sends signal $a$ under prior $\mu$: $ \pi^*(a) = \sum_{\omega \in \Omega} \mu(a) \pi^*(a | \omega)$.  Let $\mu_a \in \Delta(\Omega)$ be the posterior belief induced by signal $a$ under prior $\mu$: 
\begin{align*}
    \mu_a(\omega) = \frac{\mu(\omega) \pi^*(a|\omega)}{\pi^*(a)}, \quad \forall \omega \in \Omega. 
\end{align*}
Since $\pi^*$ is persuasive (the constraint in the linear program), $a$ must be an optimal action for the receiver on posterior $\mu_a$:  
\begin{align*}
    \E_{\omega \sim \mu_a}[v(a, \omega) - v(a', \omega)] \ge 0, ~ \forall a'\ne a.   
\end{align*}
According to inducibility assumption, there exists a belief $\eta_a \in \Delta(\Omega)$ for which $\E_{\omega \sim \eta_a}[v(a, \omega) - v(a', \omega)] \ge D > 0$ for every $a' \ne a$.  Consider the convex combination of $\mu_a$ and $\eta_a$ with coefficients $1-\delta, \delta$ (we will choose $\delta$ in the end): $\xi_a = (1 - \delta) \mu_a + \delta \eta_a$. By the linearity of expectation, $a$ must be better than any other action $a'$ by $\delta D$ on belief $\xi_a$: 
\begin{align}\label{eq:delta-D}
\E_{\xi_a}[v(a, \omega) - v(a', \omega)] = (1-\delta)\E_{\hat \mu_a}[v(a, \omega) - v(a', \omega)] + \delta \E_{\eta_a}[v(a, \omega) - v(a', \omega)] \ge \delta D. 
\end{align}
Let $\xi = \sum_{a\in A} \pi^*(a) \xi_a \in \Delta(\Omega)$, and write $\mu$ as the convex combination of $\xi$ and another belief $\chi \in \Delta(\Omega)$: 
\begin{align} \label{eq:convex-combination-2}
    \mu ~ = ~ (1-y) \xi + y \chi ~ = ~ \sum_{a\in A} (1-y) \pi^*(a) \xi_a ~ + ~ y \chi.
\end{align}

\begin{lemma}[Proposition 1 of \citet{zu_learning_2021}]
\label{lem:small-y}
If $\delta \le p_0$, then there exist $\chi$ on the boundary of $\Delta(\Omega)$ and $0\le y \le \frac{\delta}{p_0} \le 1$ that satisfy \eqref{eq:convex-combination-2}. 
\end{lemma}

Since \eqref{eq:convex-combination-2} is a convex decomposition of the prior $\mu$, according to \citep{kamenica2011bayesian}, there exists a signaling scheme $\tilde \pi$ that induces posterior $\xi_a$ with probability $(1-y) \pi^*(a)$, for $a\in \A$, and the posterior that puts all probability on $\omega$ with probability $y \chi(\omega)$, for $\omega \in \Omega$.  Namely, $\tilde \pi$ has signal space $\Ss = \A \cup \Omega$ and signal probability
\begin{align*}
    \tilde \pi(s | \omega) = \begin{cases}
    \frac{(1-y)\pi^*(a) \xi_a(\omega)}{\mu(\omega)} & \text{ for } s = a \in \A; \\
    \frac{y\chi(\omega)}{\mu(\omega)} & \text{ for } s = \omega \in \Omega; \\
    0 & \text{ otherwise.}
    \end{cases}
\end{align*}
It is not hard to verify that, under prior $\mu$ and signaling scheme $\tilde \pi$, the posterior induced by signal $a \in \A$ is equal to $\xi_a$, and the posterior induced by signal $\omega$ is the deterministic distribution on $\omega$. 

We show that, whenever $\tilde \pi$ sends an action recommendation $a \in \A$, the recommendation is persuasive for the receiver under any prior $\mu'$ in $B_1(\mu, \eps) = \{ \mu' : \| \mu' - \mu \|_1 \le \eps \}$.
\begin{claim}\label{claim:tilde-pi-persuasive}
Suppose $\delta \ge \frac{2\eps}{p_0D}$. Then, for any prior $\mu' \in B_1(\hat \mu, \eps)$, any action recommendation $a\in \A$ from $\tilde \pi$ is persuasive. 
\end{claim}
\begin{proof}
By continuity of posterior (Lemma~\ref{lem:continuity-posterior}), the posteriors induced by signal $a$ under prior $\mu$ and $\mu'$ satisfy
\begin{align*}
    \| \mu_a - \mu_a' \|_1 \le \tfrac{2}{p_0} \| \mu - \mu' \|_1 \le \tfrac{2\eps}{p_0}. 
\end{align*}
Note that the posterior $\mu_a = \xi_a$, so $\| \xi_a - \mu'_a \|_1 \le \frac{2\eps}{p_0}$.  Then, since the receiver's utility is in $[0, 1]$, for any action $a' \ne a$,
\begin{align*}
    \big| \E_{\omega \sim \mu'_a}[v(a, \omega) - v(a', \omega)] - \E_{\omega \sim \xi_a}[v(a, \omega) - v(a', \omega)] \big| \le \| \mu_a - \xi_a \|_1 \le \tfrac{2\eps}{p_0}.
\end{align*}
Together with \eqref{eq:delta-D}, we get
\begin{align*}
    \E_{\omega \sim \mu_a}[v(a, \omega) - v(a', \omega)] \ge \delta D - \tfrac{2\eps}{p_0} \ge 0. 
\end{align*}
Thus, the action recommendation $a$ is persuasive. 
\end{proof}

Then, we show that the signaling scheme $\tilde \pi$ is ``close to'' $\pi^*$ in the following sense: 
\begin{claim}
\label{claim:signaling-scheme-close}
For any $a \in \A$ and $\omega \in \Omega$, $|\tilde \pi(a|\omega) - \pi^*(a|\omega) | \le \frac{\delta}{p_0} + y$.     
\end{claim}
\begin{proof}
By definition, 
\begin{align*}
    |\tilde \pi(a|\omega) - \pi^*(a|\omega) | & ~ = ~ \Big| \frac{(1-y)\pi^*(a) \xi_a(\omega)}{\mu(\omega)} - \frac{\pi^*(a) \mu_a(\omega)}{\mu(\omega)} \Big| \\
    & ~ \le ~ (1-y) \Big| \frac{\pi^*(a) \xi_a(\omega)}{\mu(\omega)} - \frac{\pi^*(a) \mu_a(\omega)}{\mu(\omega)} \Big| + y \cdot \frac{ \pi^*(a) \mu_a(\omega)}{\mu(\omega)} \\
    & ~ = ~  (1-y) \frac{\pi^*(a)}{\mu(\omega)} \big| \xi_a(\omega) - \mu_a(\omega) \big| + y \cdot \pi^*(a|\omega) \\
    & ~ = ~  (1-y) \frac{\pi^*(a)}{\mu(\omega)} \cdot \delta \big| \eta_a(\omega) - \mu_a(\omega) \big| + y \cdot \pi^*(a|\omega) \\
    & ~ \le ~  (1-y) \frac{1}{p_0} \cdot \delta \cdot 1 + y \cdot 1  \le ~ \frac{\delta}{p_0} + y. 
\end{align*}
\end{proof}

Let $U(\mu, \tilde \pi)$ be the sender's expected utility when using signaling scheme $\tilde \pi$. Since the action recommendation from $\tilde \pi$ are persuasive under prior $\mu$ (Claim~\ref{claim:tilde-pi-persuasive}), the receiver takes $a$ when receiving signal $a$.  When receiving any signal $\omega$, the receiver takes some action $a^*_\omega \in \argmax_{a\in \A} v(a, \omega)$. So, 
\begin{align*}
U(\mu, \tilde \pi) & ~ = ~ \sum_{\omega \in \Omega} \mu_0(\omega) \Big( \sum_{a\in \A} \tilde \pi(a | \omega) u(a, \omega) ~ + ~ \tilde \pi(\omega | \omega) u(a^*_\omega, \omega) \Big). 
\end{align*}
Because we assumed $u(a, \omega) \ge 0$, 
\begin{align*}
U(\mu, \tilde \pi) & ~ \ge ~ \sum_{\omega \in \Omega} \mu_0(\omega) \sum_{a\in \A} \tilde \pi(a | \omega) u(a, \omega) ~ =: ~ U_{\A}(\tilde \pi). 
\end{align*}
where $U_\A(\tilde \pi)$ denotes the expected utility from action recommendation signals, which is also the objective function of the linear program in the definition in $U^*(\mu)$.  Note that $U_\A(\pi^*) = U^*(\mu)$.  We claim that $U_\A(\tilde \pi)$ cannot be too much worse than $U_\A(\pi^*)$: 
\begin{claim}\label{claim:tilde-pi-approximately-optimal}
Given $\delta \ge \tfrac{2\eps}{p_0D}$, we have 
$U_\A(\tilde \pi) \ge U_\A(\pi^*) - \tfrac{2\delta}{p_0}$. 
\end{claim}
\begin{proof}
By definition, 
\begin{align*}
    U_\A(\tilde \pi) & = \sum_{\omega \in \Omega} \mu_0(\omega) \sum_{a\in \A} \tilde \pi(a | \omega) u(a, \omega) \\
    \text{(by Claim~\ref{claim:signaling-scheme-close})} & \ge \sum_{\omega \in \Omega} \mu_0(\omega) \sum_{a\in \A} \pi^*(a | \omega) u(a, \omega) - \Big(\frac{\delta}{p_0} + y \Big) \underbrace{\sum_{\omega \in \Omega} \mu_0(\omega) \sum_{a\in \A} u(a, \omega)}_{\le 1} \\
    & \ge U_\A(\pi^*) - y - \frac{\delta}{p_0} \\
    & \ge U_\A(\pi^*)  - \frac{2\delta}{p_0}
\end{align*}
where the last line used $y \le \frac{\delta}{p_0}$ from Lemma~\ref{lem:small-y}. 
\end{proof}

Because $\tilde \pi$ is persuasive for any prior $\mu' \in B_1(\mu, \eps)$ and $U_\A(\tilde \pi) \ge U_\A(\pi^*) - \tfrac{2\delta}{p_0}$, we have:
\begin{align*}
    U^*(\mu')  \ge ~ U(\mu', \tilde \pi) &\ge ~ U_\A(\tilde \pi) \\
              & ~ \ge ~ U_\A(\pi^*) - \frac{2\delta}{p_0} \\
              & ~ = ~ U^*(\mu) - \frac{2\delta}{p_0}  \ge ~ U^*(\mu) - \frac{4\eps}{p_0^2 D}
\end{align*}
where we let $\delta = \frac{2\eps}{p_0 D}$. By a symmetric argument, we also have $U^*(\mu)  \ge ~ U^*(\mu') - \frac{4\eps}{p_0^2 D}$, which implies $|U^*(\mu') - U^*(\mu)| \le \frac{4\eps}{p_0^2 D}$.

\subsection{Proof of Theorem \ref{thm:joint-unconstrained}}
\label{proof:joint-unconstrained}
\begin{proof}
    Beginning with the first claim, recall that we consider sender utilities to be positive (this is without loss of generality since the sender utility is linear in $u(a, \omega)$, allowing us to normalize as needed). Let $a^u(\omega) = \argmax_{a}{u(a, \omega)}$ and $a^v(\omega) = \argmax_{a}{v(a, \omega)}$ denote the optimal action for the sender and receiver at state $\omega$ respectively. Further, let $\omega_{\min} = \argmin_{\omega}{\mu_0(\omega)u(a^u(\omega), \omega)}$; it captures the ``least'' important state for the sender assuming sender-optimal action at each state.

    Since $B = \Delta(\Omega)$, consider using framing to induce a belief $\mu_c(\omega_{\min}) = 1 - \varepsilon$ and $\frac{\varepsilon}{|\Omega|-1}$ for all other states; let the signaling scheme $\pi$ deterministically recommend the received optimal action at $\omega_{\min}$ and sender optimal actions at all other states. In other words:
     \begin{equation*}
        \pi(a^v(\omega_{\min}) | \omega_{\min}) = 1 \quad \text{and} \quad \forall \omega \ne \omega_{\min} \, : \, \pi(a^u(\omega) | \omega) = 1. 
    \end{equation*}
    Since the sender's utility is non-negative, if the receiver follows the recommended actions outlined by this scheme, the sender is guaranteed to achieve at least the following utility:
    \begin{equation*}
        \sum_{\omega \ne \omega_{\min}}{\mu_0(\omega)u(a^u(\omega), \omega)} \geq u_{\max} - \frac{1}{|\Omega|}u_{\max} = \left(1 - \frac{1}{|\Omega|}\right)u_{\max}
    \end{equation*}
    where we note that the maximal possible utility achievable by the sender is $u_{\max} = \sum_{\omega}{\mu_0(\omega)u(a^u(\omega), \omega)}$ and by the pigeonhole principle, $\mu_0(\omega_{\min})u(a^u(\omega_{\min}), \omega_{\min}) \le \tfrac{1}{|\Omega|}u_{\max}$. We now show that following the recommended actions is $\varepsilon$-obedient for the receiver. Indeed, for a recommended action $a$ and any other action $a'$, the $\eps$-obedience expression for this pair under the scheme $\pi$ is:
    \begin{equation*}
        (1-\varepsilon)\pi(a|\omega_{\min})[v(a, \omega_{\min}) - v(a', \omega_{\min})] + \sum_{\omega \ne \omega_{\min}}{\frac{\varepsilon}{|\Omega|-1}\pi(a|\omega)[v(a, \omega) - v(a', \omega)]} ~ \ge ~ -\eps. 
    \end{equation*}
    When the receiver gets recommended action $a^v(\omega_{\min})$, this expression becomes at least $(1-\varepsilon)$ $[v(a^v(\omega_{\min}), \omega_{\min})- v(a',\omega_{\min} )] - \varepsilon$ $ \geq -\varepsilon$ since at state $\omega_{\min}$, action $a^v(\omega_{\min})$ is optimal for the receiver. Conversely, if the receiver is recommended some action $a \ne a^v(\omega_{\min})$, then the expression is: $\sum_{\omega \ne \omega_{\min}}$ ${\tfrac{\varepsilon}{|\Omega|-1}\pi(a|\omega)[v(a, \omega) - v(a', \omega)]}$ $ \geq -\varepsilon$ since the utilities are bounded to $[0,1]$. So, $\eps$-obedience is satisfied.   

    For the second claim, we can make use of an additional assumption: the sender utility is state-independent. This means that their utility is maximized if the receiver takes some action $a^u$ at \emph{all} possible states. We also recall from Section~\ref{sec:model} that any action of the receiver is inducible -- that is, for any action $a \in \A$, there exists some belief $\mu_a$ wherein taking $a$ is optimal for the receiver. For the sender optimal action $a^u$, let $\mu_{a^u}$ be the belief where $a^u$ is optimal for the receiver. Indeed, $\mu_{a^u}$ can be computed using the following set of linear constraints:
    \begin{equation*}
        \,\sum_{\omega}\mu(\omega)[v(a^u, \omega) - v(a', \omega)] \geq 0, \quad \forall a' \in \A.
    \end{equation*}
    Since $B = \Delta(\Omega)$, we can choose a framing $c$ to induce belief $\mu_{a^u}$. We accompany this framing with an uninformative signaling scheme $\pi$. Such a scheme reveals no information about the realized state. For example, for a given action $a_1$, $\pi(a_1| \omega) = 1$ for all $\omega$ is an uninformative scheme. Under this joint strategy, the receiver's belief is always $\mu_{a^u}$, where their best-response is to take the sender optimal action $a^u$. This is thus an optimal strategy for the sender.  
 \end{proof}


\section{Experimental Setup: Advertising Case Study}\label{appendix:advertising_exp}
\subsection{Instance Setup and Parameters}
Here we include the detailed setup of 
the advertising example
we used to verify our framework experimentally. We describe the target customer and brand to the LLM as follows. These are used to generate new framing 
as well as estimate customers' priors. 
The description of the outerwear line are taken directly from 
the product line of a real-world clothing brand (Patagonia), with the brand name changed to ``Himalaya''. 
Note that the utilities here are not within the range $[0,1]$, but can be normalized to be so without loss of generality. 

\begin{itemize}
\item States: There are 4 possible states: \emph{(Chic/Trendy, Durable), (Chic/Trendy, Not Durable), (Functional/Practical, Durable), (Functional/Practical, Not Durable)}.
\item Description of Target Demographic: \emph{Our new target demographic are fashion-aware average mall consumers, who are casually into an active lifestyle. They are middle class, but are willing to pay a slight premium for quality and style. This is a segment that the athleisure market has dominated of late, with brands like Lululemon, Nike being the main players. Consumers here like the idea of performance wear (e.g., for hiking or skiing) but are not deeply familiar with or motivated by technical characteristics. What matters most is whether the clothes looks stylish in everyday environments like schools, cafés, or city streets. Functionality and durability is a nice bonus, but aesthetic appeal primarily drives their interest.}
\item Description of the Brand: \emph{Himalaya is fairly well known brand in the outdoor enthusiast, mountaineering, and adventure community. It has a reputation for bulletproof build quality and performance, and valuing sustainability. It is launching a new outerwear line, that includes parkas, ski-jackets and ski-pants, windbreakers, and thermal layers. All products here are made with 100\% postconsumer recycled nylon ripstop and without PFAS. They meet H2No Performance Standard for waterproofness and breathability. Fabric and inner membrane have durable water repellent (DWR) finish.}
\item True/Sender/Brand Prior: $\mu_0 = [0.225, 0.125, 0.5, 0.15]$. This was obtained by asking the LLM about how products would fit into each of the categories (i.e. their prior) given the brand description above.
\item Brand Utility (columns correspond to the 4 states and the rows are "buy on sale", "buy regular price" and "not buy" actions:
\begin{equation*}
\begin{bmatrix}
    \text{N/A} & 4.0 & 5.0 & 4.5 \\ 
    8.0 & 7.0 & 6.5 & 6.0 \\ 
    0 & 0 & 0 & 0
\end{bmatrix} 
\end{equation*}
\item Target Consumer Utility (columns correspond to the 4 states and the rows are "buy on sale", "buy regular price" and "not buy" actions:
\begin{equation*}
\begin{bmatrix}
    \text{N/A} & 5.0 & -30 & -40 \\ 6.0 & 4.0 & -50 & -65 \\ 0 & 0 & 0 & 0
\end{bmatrix} 
\end{equation*}

\end{itemize}

\subsection{Prompts Used for LLMs}\label{app_sec:llm_prompts}
\paragraph{\bf Estimating framing induced beliefs: } To estimate the belief for a given framing, we use the following prompt template. The strings \texttt{Target Demographic} and \texttt{Brand Description} are given above.

{\itshape
You will be used as a proxy for a target demographic to assess shopping inclinations for market research. You will be given a description of the demographic (their preferences, etc) and the motto and description of a clothing line. You will be asked to provide your responses in a JSON format specified in the prompt

GENERAL PROBLEM DESCRIPTION: You are taking the role of someone in the given demographic. You can imagine they categorize clothing into the following categories: (trendy, more durable), (trendy, less durable), (not trendy, more durable), (not trendy, less durable). Please see BUYER DESC for what this buyer values.

BUYER DESC: \texttt{Target Demographic}

BRAND DESC: \texttt{Brand Description}

TASK: Your role is to act as a member of the described demographic and evaluate how you would interpret the products from a new outerwear line based solely on the brand’s motto and product description (See BRAND DESC). Your goal is to determine how you (as an average style-aware mall-shopper) would categorize the product line into the following four quadrants: (Chic/Trendy, Durable), (Chic/Trendy, Not Durable), (Functional/Practical, Durable) and (Functional/Practical, Not Durable). Please return what probabilities (recall they sum to 1) this average user from this demographic would assign to each of these categories for products from this line. Note that this is not about what the demographic cares about or prioritizes in purchases. Instead, focus on how they would interpret the messaging — what assumptions they would make about the clothing’s fashionability and durability from the language, tone, and emphasis in the brand’s description and motto. IMPORTANT: DO NOT MAKE FAR REACHING ASSUMPTIONS OR TRY TO BE UNJUSTIFIABILY AGGRESSIVE. Provide your response in the following JSON format: 

\begin{verbatim}
{
    "reasoning": string,
    "probabilities": {
        "trendy_more_durable": float,
        "trendy_less_durable": float,
        "not_trendy_more_durable": float,
        "not_trendy_less_durable": float
    },  
}
\end{verbatim}
}

\paragraph{\bf Optimizing over the framing space: } To search over the framing space, we use the following prompt template. The key-words \texttt{BRAND\_DESC} and \texttt{BUYER\_DESC} correspond to the instance parameters mentioned earlier. Any generated framing and corresponding feedback is appended to this prompt for the next iteration:

{\itshape
You will be asked to generate a brand motto and description for one of its product lines. For each motto, description you generate, quantitative feedback will be provided on the generated, which you will use to improve what you generate.

TASK DESC: You will be given a BRAND\_DESC that describes the clothing brand 'Himalaya' and a new product line they are trying to launch. You will be given BUYER\_DESC that outlines the features of the demographic they are targeting for this new product line. Your task is to generate a BRAND MOTTO (at most 10 words) and PRODUCT LINE DESC of their new product line (at most 100 words). The motto and description will be shown to members in the target demographic. Their perception of how products from this new line fit into the 4 possible categories this demographic cares about will be measures (quantitatively). Please see BUYER\_DESC for how they partition clothes into 4 possible states - it is their belief over these states that we measure. Feel free to USE OR NOT USE any information in the provided BRAND DESC to sway the target demographic. Not revealing information can sometimes be helpful. Using this perceived prior, we will signal the buyer (using Bayesian Persuasion) to influence their actions (which buy-on-sale, buy-regular-price, not buy). We will compute all of this and give you the numerical utility the company achieves when using your generated motto and description. See FEEDBACK DESC on how the feedback will be structured. Please use this feedback to improve the BRAND MOTTO and PRODUCT DESC you generate. Note that your generated motto and description directly influences how the buyer perceives this Himalaya's new product line. This is captured in their prior, which will directly influence the utility we can derive after signaling. Lastly, feel free to navigate this space to see what works and what doesn't.

FEEDBACK\_DESC: After each BRAND\_MOTTO and PRODUCT\_LINE\_DESC you generate, you will receive feedback: 

1) Current Belief: The belief induced by your framing and the consumer's reasoning.

2) Current Utility, Correctness and Language score: The expected utility the brand achieves with this belief and your scores on not using false claims and use of good natural language.

3) Total Score: (average of quality scores) * utility. THIS IS WHAT WE MAXIMIZE.

4) Belief Change: How beliefs changed from the previous round and whether utility improved or decreased.

5) Best So Far: The best framing found so far with its belief and total score for reference.

IMPORTANT: Use natural, common language.
Do not fabricate claims not supported by the brand facts."

Please generate a BRAND\_MOTTO and PRODUCT\_LINE\_DESC in json form with those as keys. Avoid just re-stating the preferences of the buyer.
}

\paragraph{\bf Verifying soundness of framing: } We use LLMs to also verify the framings generated to ensure they are not making incorrect claims or using poor language.

{\itshape

You will be given factual information about a clothing brand (labelled BRAND\_DESC) and a generated
marketing description (labelled GENERATED\_DESC) that will be shown to consumers.

Your goal is to score the GENERATED\_DESC on factual correctness using ONLY these three scores: 1.0, 0.5, or 0.0.

CORRECTNESS CRITERIA:
- The description should NOT contain claims that directly contradict the BRAND\_DESC

- The description should NOT fabricate specific claims not supported by BRAND\_DESC (e.g., claiming "award-winning" when not mentioned)

- Omitting information is FINE and does not violate correctness

- General marketing language and subjective claims (e.g., "stylish") are FINE

- Specific technical claims must align with BRAND\_DESC if made

SCORING GUIDE (use ONLY these values):
- 1.0: Minor issues or no issues. No clear factual contradictions. Plausible marketing claims are acceptable even if not explicitly stated in BRAND\_DESC. Be generous here.

- 0.5: Some claims are erroneous or stretch the facts. Contains unsupported specific claims that go beyond reasonable marketing language.

- 0.0: Major errors. Contains clear factual lies or direct contradictions of BRAND\_DESC.

Return JSON with keys: "reasoning" (string explaining your score) and "correctness\_score" (must be exactly 1.0, 0.5, or 0.0).
}

\subsection{Framings}\label{app_sec:optimal_framings}

The following is the Patagonia framing (their actual motto and description of a product line) used for the belief estimation experiments and as the starting point for the optimal framing search. Thereafter, we give the best framings generated by Claude Sonnet 4 and GPT 5.2.

\noindent \textbf{Patagonia Slogan:} Build the best product.

\noindent \textbf{Patagonia Product Line Description: } Whether you're seeking boulders and alpine views or stalking rainbow trout, this low-profile, waterproof jacket will keep you dry and provide rain protection from the fork in the trail to the places less traveled, and back again. This 3-layer shell meets our H2No® Performance Standard for exceptional waterproof/breathable protection, and the fabric, membrane and durable water repellent (DWR) finish are made without intentionally added PFAS.

\noindent \textbf{Claude Optimal Slogan:} City Sharp. Mountain Ready.

\noindent \textbf{Claude Product Line Description: } Our newest collection redefines modern outerwear with clean lines, contemporary cuts, and versatile designs that transition effortlessly from city to countryside. Each piece features premium sustainable fabrics and thoughtful construction details that ensure lasting quality without compromising on style. Perfect for the modern lifestyle - whether you're commuting to work, meeting friends for brunch, or enjoying weekend adventures. These aren't just clothes, they're wardrobe essentials that reflect your values and elevate your everyday look with confidence and sophistication.

\noindent \textbf{GPT Optimal Slogan:} Where style meets substance.

\noindent \textbf{GPT Product Line Description: } Himalaya’s new outerwear brings modern, everyday style to real weather: streamlined parkas, ski jackets and pants, windbreakers, and thermal layers that look clean on campus, in cafés, and on the go. The fits are sleek and easy to layer—made to pair with denim, joggers, and sneakers. Underneath the minimal finish: 100\% recycled nylon ripstop, a durable water‑repellent finish, and high performance waterproof/breathable comfort for rainy days and weekend trips.

\end{document}